\documentclass[10pt]{article} 
\usepackage[preprint]{tmlr}

\usepackage{hyperref}
\usepackage{url}

\title{Delaunay--Rips filtration: a study and an algorithm}

\author{\name Mattéo Clémot \email matteo.clemot@univ-lyon1.fr \\
      \addr LIRIS, Université Claude Bernard Lyon 1
      \AND
      \name Julie Digne \email julie.digne@cnrs.fr \\
      \addr CNRS\\
      LIRIS, Université Claude Bernard Lyon 1
      \AND
      \name Julien Tierny \email julien.tierny@sorbonne-universite.fr\\
      \addr CNRS\\
      LIP6, Sorbonne University}

\def\month{07}  
\def\year{2026} 
\def\openreview{\url{https://openreview.net/forum?id=XXXX}} 

\usepackage{amsmath,amssymb,amsfonts}
\usepackage{amsthm}
\usepackage{stmaryrd}
\usepackage[ruled,vlined,linesnumbered]{algorithm2e}
\usepackage[nameinlink,capitalise]{cleveref} \creflabelformat{equation}{#2#1#3}
\usepackage{multirow}
\usepackage{tikz-cd}
\usepackage{wrapfig}

\usepackage{soul}
\usepackage{xcolor}

\definecolor{customRed}{HTML}{cf244b}
\definecolor{customBlue}{HTML}{493df5}

\newtheorem{theorem}{Theorem}
\newtheorem{lemma}{Lemma}

\graphicspath{{Figures/}}

\newcommand{\ddim}{d}
\newcommand{\hdim}{k}

\newcommand{\pc}{X}
\newcommand{\jung}[1]{j_{#1}}

\newcommand{\mst}{\ensuremath{\mathsf{MST}}}
\newcommand{\rng}{\ensuremath{\mathsf{RNG}}}
\newcommand{\ug} {\ensuremath{\mathsf{UG}}}
\newcommand{\del}{\ensuremath{\mathsf{DEL}}}

\newcommand{\cech}{\v{C}ech}

\newcommand{\cechn}{\mathcal{\check{C}}}
\newcommand{\delcechn}{\mathcal{D\check{C}}}
\newcommand{\ripsn}{\mathcal{R}}
\newcommand{\delripsn}{\mathcal{DR}}

\newcommand{\offsetf}[1]{\mathcal{O}_{#1}}
\newcommand{\cechf}[1]{\mathcal{\check{C}}_{#1}}
\newcommand{\delcechf}[1]{\mathcal{D\check{C}}_{#1}}
\newcommand{\alphaf}[1]{\mathcal{D}_{#1}}
\newcommand{\ripsf}[1]{\mathcal{R}_{#1}}
\newcommand{\delripsf}[1]{\mathcal{DR}_{#1}}

\DeclareMathOperator{\vor}{Vor}
\DeclareMathOperator{\meb}{meb}

\newcommand{\bbr}{\mathbb{R}}
\newcommand{\bbs}{\mathbb{S}}
\newcommand{\bbv}{\mathbb{V}}
\newcommand{\bbw}{\mathbb{W}}

\newcommand{\calf}{\mathcal{F}}

\newcommand{\calk}{\mathcal{K}}

\newcommand{\calo}{\mathcal{O}}

\newcommand{\calw}{\mathcal{W}}

\DeclareMathOperator{\nrv}{nrv}
\DeclareMathOperator{\bottle}{\mathrm{d_b}}
\DeclareMathOperator{\dgh}{\mathrm{d_{GH}}}
\DeclareMathOperator{\hausdo}{\mathrm{d_{H}}}

\newcommand{\ratio}{\eta}

\newcommand{\dgm}[1]{\mathrm{dgm}^{#1}}
\newcommand{\dgmcech}[1]{\mathrm{dgm}_{\mathcal{\check{C}}}^{#1}}
\newcommand{\dgmrips}[1]{\mathrm{dgm}_{\mathcal{R}}^{#1}}
\newcommand{\dgmdr}[1]{\mathrm{dgm}_{\mathcal{DR}}^{#1}}

\newcommand{\augmented}[1]{\overline{#1}}

\newcommand{\hgf}[3]{H_{#1}({#2}_{#3})} 
\newcommand{\hgfx}[4]{H_{#1}({#2}_{#3}({#4}))} 
\newcommand{\hgdr}[4]{V_{#3}(#4)}

\newcommand{\maprtodr}[2]{\phi_{#1}^{#2}}
\newcommand{\mapxtoy}[2]{\phi_{#1}}
\newcommand{\mapytox}[2]{\psi_{#1}}
\newcommand{\mapgeoxtoy}[2]{\alpha_{#1}}
\newcommand{\mapgeoytox}[2]{\beta_{#1}}

\newcommand{\betti}[1]{\beta_{#1}}

\newcommand{\skeleton}[2]{{#1}^{({#2})}}

\newcommand{\PH}[1]{\ensuremath{\text{PH}^{#1}}}

\newcommand{\usc}[1]{\mathsf{US}^{#1}}
\newcommand{\rnsc}[1]{\mathsf{RNS}^{#1}}
\newcommand{\msa}[1]{\mathsf{MSA}^{#1}}
\newcommand{\nmj}[1]{\mathsf{NM}^{#1}}

\newcommand{\poly}[1]{\mathsf{cells}^{#1}}

\newcommand{\dist}[2]{\mathrm{d}({#1},{#2})}
\DeclareMathOperator{\cof}{cof}

\newcommand{\lex}{\prec_\mathrm{lex}}
\newcommand{\order}[1]{\prec_{#1}}

\newcounter{fnsave}

\newcommand{\figfootnotemark}{%
  \refstepcounter{footnote}%
  \hyperlink{figfn\thefootnote}{\textsuperscript{\thefootnote}}%
}

\newcommand{\figfootnotetext}[1]{%
  \stepcounter{footnote}%
  \hypertarget{figfn\thefootnote}{}\footnotetext{#1}%
}

\newcommand{\vertical}[1]{\rotatebox{90}{#1}}

\begin{document}

\maketitle

\begin{abstract}
      The Delaunay--Rips filtration is a lighter and faster alternative to the well-known Rips filtration for low-dimensional Euclidean point clouds, making it a promising tool for topological approaches in machine learning. However, it has seldom been studied.
      In this paper, we address this gap by providing a comprehensive theoretical and empirical analysis of this construction, as well as an evaluation of its practical performance in classification tasks.
      From a theoretical perspective, we show how the persistence diagrams associated with the Delaunay--Rips filtration approximate those obtained with the Rips filtration and we describe the instabilities of the Delaunay--Rips persistence diagrams when the input point cloud is perturbed.
	Then, we introduce an algorithm that computes persistence diagrams of Delaunay--Rips filtrations in any dimension.
	We show that our approach is faster and has a lower memory footprint than traditional approaches in low dimensions.
      Finally, we evaluate its performance in classification tasks on synthetic and real datasets.
	Our C++ implementation, which comes with Python bindings, is available at \url{https://github.com/MClemot/GeoPH}.
\end{abstract}

\section{Introduction}
\label{sec:introduction}

Topological data analysis (TDA) provides various tools to handle topological properties of datasets. In particular, it allows us to perform \textit{topological inference} of point clouds, i.e., finding whether they feature interesting topological structures like connected components, cycles, or cavities, and have a measure of their \emph{significance}. To do so, TDA relies on \emph{filtrations}, which are increasing sequences of simplicial complexes over a set of points. One of them, the (Vietoris)--Rips filtration, is a particularly versatile tool, as it can be applied in any metric space, and in particular in any Euclidean space $\bbr^\ddim$.
Unfortunately, the Rips filtration becomes prohibitively large as the number $n$ of points increases, since the total number of simplices to consider is $2^n-1$. Even its $\hdim$-skeleton, i.e., the set of simplices with dimension at most $\hdim$, has $\calo(n^{\hdim+1})$ simplices. Therefore, despite the availability of optimized implementations~\cite{bauer_ripser_2021}---including parallelization on the CPU~\cite{perez2021giotto} or the GPU~\cite{zhang2020gpu}---using the Rips filtration can be very computationally expensive, both in time and memory.

In low dimensional Euclidean spaces, one can leverage the Delaunay complex in order to build filtrations with a drastically reduced number of simplices. This is the case of the well-known $\alpha$-filtration (i.e., the filtration of $\alpha$-complexes, originally introduced in \cite{edelsbrunner1994three}), also known as the Delaunay filtration~\cite{bauer2017morse}. It has the nice property of being topologically equivalent to the less combinatorial offset filtration, i.e., the filtration of the union of the increasing balls centered at the input points. On the contrary, the Rips filtration only provides an approximation of this offset filtration.

A less common way of constructing a filtration is to restrict the Rips filtration to the Delaunay complex, which we call here the \emph{Delaunay--Rips filtration}, but which is also known as the \emph{weak $\alpha$-filtration}. It has been recently used for constructing a differentiable topological layer, for instance to incorporate topological priors in learning processes~\cite{gabrielsson2020topology}. It was also compared to $\alpha$-filtrations and Rips filtrations in classification tasks where it performed similarly in terms of accuracy~\cite{mishra2023stability}.
In addition, while the Delaunay--Rips filtration, like the Rips filtration and contrary to the $\alpha$-filtration, does not benefit from the topological equivalence to the offset filtration, it is faster to compute than the $\alpha$-filtration since it relies on simplex diameters instead of (empty) minimum enclosing ball radii, which are more computationally intensive. Furthermore, its flag complex structure can be leveraged when computing its persistent homology.%

To our knowledge, the Delaunay--Rips filtration and its persistent homology have seldom been studied and very few practical implementations are available (giotto-tda~\cite{JMLR:v22:20-325} proposes an implementation based on Ripser~\cite{bauer_ripser_2021}).
In this paper, we aim to study this filtration and introduce a dedicated algorithm for computing its persistent homology. More precisely, we make the following contributions:
\begin{enumerate}
    \item a theoretical and empirical study of the instabilities of Delaunay--Rips persistence diagrams and their distance to Rips persistence diagrams (\cref{sec:analysis});
    \item a dedicated algorithm to compute Delaunay--Rips persistence diagrams in any dimension (\cref{sec:method});
    \item an evaluation of its performance in classification tasks involving synthetic and real datasets compared to other filtrations (\cref{sec:classification});
    \item a C++ implementation\footnote{available at this repository: \url{https://github.com/MClemot/GeoPH}} (with Python bindings\footnote{available on PyPI: \url{https://pypi.org/project/geoph/}}) of this algorithm.
\end{enumerate}

\section{Related work}

Much work has been done to improve the use of the Rips filtration. Algorithmically, it is common to compute cohomology instead of homology, which produces the same persistence diagrams. Ripser~\cite{bauer_ripser_2021} specializes to Rips filtrations by representing implicitly the coboundary matrix and by dealing efficiently with the so-called \textit{apparent} and \textit{emergent pairs}. Concerning parallelism, the Rips filtration construction and the finding of the apparent pairs have been implemented on the GPU by Ripser++~\cite{zhang2020gpu}, but not the matrix reduction step. A parallel implementation of this step on the CPU was proposed by giotto-ph~\cite{perez2021giotto}, following the lock-free shared-memory approach of~\cite{morozov2020towards}. 
However, even with these parallel approaches, and even with a limitation on the dimension of the simplices and their value of the filtration, Rips and \cech\ filtrations remain large and become prohibitively costly to compute as the number $n$ of points increases.

\paragraph{Reducing the size of filtrations}

Linear-size filtrations depending on an approximating factor $\varepsilon$ have been proposed to provide $(1+\calo(\varepsilon))$-approximations of the Rips or the \cech\ filtrations~\cite{sheehy2012linear, kerber2013approximate, cavanna2015geometric}. Another way to reduce flag filtrations (that are made of flag complexes, like the Rips complexes) is to use \emph{edge collapses}~\cite{boissonnat2020edge, glisse2022swap} that return a smaller filtration with the same persistent homology, thereby speeding up its computation.
Finally, other methods select a reduced set of landmarks to build smaller filtrations that also provide a good estimate of the topology~\cite{de2004topological,graf2025floodcomplex}.

\paragraph{Learning-based approaches}
Another way to avoid the expensive computation of persistence diagrams of the Rips filtration consists in using a neural network that was previously trained to map point sets to some vectorization of these diagrams. For example, TopologyNet~\cite{zhou2022learning} and RipsNet~\cite{de2022ripsnet} propose architectures (based respectively on the EdgeConv~\cite{wang2019dynamic} and DeepSets~\cite{zaheer2017deep} architectures) that can learn finite-dimensional embeddings of Rips persistence diagrams.

\paragraph{Computational optimizations in low dimension}

Some optimizations are possible when considering point clouds in Euclidean spaces (all the more so in low dimension). In particular, the \textit{reduced} Vietoris--Rips filtration~\cite{koyama2023faster} has the same 1-dimensional persistent homology (\PH{1}) as the Rips filtration but contains only $\calo(n^2)$ instead of $\calo(n^3)$ triangles, relying on the relative neighborhood graph (see \cref{sec:geometric-structures}). This enables a faster and more memory-efficient computation of the \PH{1} of the Rips filtration of low-dimensional point clouds. In the Euclidean plane, minmax length triangulations can be leveraged to further speed up the search for triangles that destroy a \PH{1} class~\cite{clemot2025topological}.
Besides, working in a low-dimensional Euclidean space allows leveraging Delaunay
complexes to build filtrations (e.g., Delaunay and Delaunay--Rips filtrations)
with a drastically reduced number of simplices, while still providing
a good approximation of the offset filtration.

\paragraph{Applications}

The Delaunay--Rips filtration was used for shape classification (based on random forests) or biophysical time series classification (based on SVMs)~\cite{mishra2023stability}. The authors claim that, in practice, the instability of the Delaunay--Rips persistence diagrams (see \cref{sec:stability}) has little impact on the accuracy of these applications (in comparison to using the Rips- or the $\alpha$-filtration).
The Delaunay--Rips filtration was also successfully used to build topological loss functions that enable regularizing or incorporating topological priors in machine learning models~\cite{gabrielsson2020topology}.

\section{Background}
\label{sec:background}

We consider a point cloud $\pc$ with $n$ points in the $\ddim$-dimensional Euclidean space $\bbr^\ddim$.
An \emph{abstract simplicial complex} is a collection $\calk$ of subsets $\sigma\subset\pc$ (called \textit{simplices}) that is closed under taking subsets (called their \textit{faces}), i.e., if $\sigma\in\calk$ and $\tau\subset\sigma$, then $\tau\in\calk$.
We denote by $\calk^\hdim=\{\sigma\in\calk\mid\dim\sigma=\hdim\}$ the set of $\hdim$-dimensional simplices of $\calk$, and $\skeleton{\calk}{k}=\{\sigma\in\calk\mid\dim\sigma\leq\hdim\}$ its $\hdim$-\emph{skeleton}. We also denote by $\Delta(X)=2^\pc\setminus\{\varnothing\}$ the full abstract simplicial complex spanned by $\pc$.

A simplicial complex of interest is the \emph{Delaunay complex} $\del(\pc)$.
It is uniquely defined when no $\ddim+1$ points lie on the same hyperplane and no $\ddim+2$ points lie on the same $(\ddim-1)$-dimensional sphere, which is referred to as the \emph{general position hypothesis}.
In the worst case, the Delaunay complex contains $\calo(n^{\lceil\ddim/2\rceil})$ simplices~\cite{amenta2007complexity}.
In practice however, one can expect a construction time complexity in $\calo(n\log n)$, with a constant exponential in the ambient dimension $\ddim$ for uniform random vertices, thanks to a pre-sorting~\cite{boissonnat2009incremental, cgal:hdj-t-24b}.

We use in the following concepts from persistent homology. We refer the reader to reference books \cite{edelsbrunner_computational_2010, zomorodian_computational_2010} for a complete introduction to simplicial and persistent homology.

\subsection{Persistence diagrams and bottleneck distance}

For a given dimension $\hdim\geq0$, the $\hdim$-dimensional \textit{persistence diagram} $\dgm{\hdim}(\calf)$ of a filtration $\calf$ is a multiset of points $(b,d)\in\bbr^2$ where $b$ and $d$ are respectively the \textit{birth} and the \textit{death} of a $\hdim$-dimensional persistence class.
Persistence diagrams can be compared with metrics borrowed from optimal transport, in particular the \emph{bottleneck distance}.
Given $D$ and $D'$ two persistence diagrams corresponding to the same persistent homology dimension, we define their \emph{augmented} diagrams $\augmented{D}=D\cup\Pi(D')$ and $\augmented{D'}=D'\cup\Pi(D)$ where $\Pi$ is the map on $\bbr^2$ that projects orthogonally on the diagonal, i.e., $\Pi:(b,d)\mapsto\left(\frac{b+d}{2}, \frac{b+d}{2}\right)$. The set of bijections between the two is denoted $\Psi=\mathrm{Bij}\left(\augmented{D},\augmented{D'}\right)$. The cost $c_\infty$ between $x\in\augmented{D}$ and $y\in\augmented{D'}$ is defined as 0 if both points belong to the diagonal and $\lVert x-y\rVert_\infty$ otherwise. This permits us to define the 
bottleneck distance:
\begin{equation}
    \bottle(D,D')=
    \min\limits_{\psi\in\Psi}\max\limits_{x\in\augmented{D}}c_\infty\bigl(x,
    \psi(x)\bigr).
\end{equation}

\subsection{Persistence modules and interleavings}
\label{sec:interleaving}

Let $\calf=(\calf_r)_{r\geq0}$ be a filtration. For all $r\leq r'$, the
inclusion map $\iota_{r,r'}:\calf_r\hookrightarrow\calf_{r'}$ induces, when
applying the homology functor, a homomorphism
$(\iota_{r,r'})_\star:\hgf{\hdim}{\calf}{r}
\rightarrow\hgf{\hdim}{\calf}{r'}$. This allows us to define the
$\hdim$-th \emph{persistence module} of $\calf$ as the collection $\left(\hgf{\hdim}{\calf}{r}\right)_{r\geq0}$ with these
induced inclusion maps $((\iota_{r,r'})_\star)_{0\leq r\leq r'}$.
Two persistence modules $\bbv=(V_r)_{r\geq0}$ and $\bbw=(W_r)_{r\geq0}$ are said to be strongly $\varepsilon$-\emph{interleaved} if there exist homomorphisms $\Phi_r:V_r\rightarrow W_{r+\varepsilon}$ and $\Psi_r:W_r\rightarrow V_{r+\varepsilon}$ for all $r\geq0$, such that the following diagrams commute for all $r\leq r'$ (where horizontal maps are the induced inclusions)~\cite{chazal2016structure}:
\begin{equation}
	\label{eq:interleaving}
	\begin{tabular}{cc}
		\begin{tikzcd}[cramped, row sep=scriptsize, column sep=scriptsize]
			V_{r-\varepsilon} \arrow[rr] \arrow[dr,"\Phi_{r-\varepsilon}"'] & & V_{r+\varepsilon} \\
			& W_{r} \arrow[ur,"\Psi_{r}"']
		\end{tikzcd}
		& 
		\begin{tikzcd}[cramped, row sep=scriptsize, column sep=scriptsize]
			& V_{r+\varepsilon} \arrow[rr] & & V_{r'+\varepsilon} \\
			W_{r} \arrow[rr] \arrow[ur,"\Psi_{r}"] & & W_{r'} \arrow[ur,"\Psi_{r'}"'] 
		\end{tikzcd}
		\\\\
		\begin{tikzcd}[cramped, row sep=scriptsize, column sep=scriptsize]
			& V_{r} \arrow[dr,"\Phi_{r}"] \\
			W_{r-\varepsilon} \arrow[rr] \arrow[ur,"\Psi_{r-\varepsilon}"] & & W_{r+\varepsilon}
		\end{tikzcd}
		&
		\begin{tikzcd}[cramped, row sep=scriptsize, column sep=scriptsize]
			V_{r} \arrow[rr] \arrow[dr,"\Phi_{r}"'] & & V_{r'} \arrow[dr,"\Phi_{r'}"] \\
			& W_{r+\varepsilon} \arrow[rr] & & W_{r'+\varepsilon}
		\end{tikzcd}
	\end{tabular}
\end{equation}
The strong stability theorem~\cite{chazal2009proximity} states that if two persistence modules $\bbv$ and $\bbw$ are strongly $\varepsilon$-interleaved, then:
\[\bottle\bigl(\dgm{}(\bbv), \dgm{}(\bbw)\bigr)\leq\varepsilon.\]
The multiplicative version of interleaving (see, e.g., \cite{sheehy2012linear}) is also of interest. In this case, we instead have homomorphisms $\Phi_r:V_r\rightarrow W_{r\ratio}$ and $\Psi_r:W_r\rightarrow V_{r\ratio}$ and the corresponding commutative diagrams.
By reparameterizing the filtrations on a log-scale, the strong stability theorem becomes in its multiplicative version that if two persistence modules are multiplicatively strongly $\ratio$-interleaved, then:
\[\bottle\bigl(\log\dgm{}(\bbv), \log\dgm{}(\bbw)\bigr)\leq\log(\ratio)\text{
with }\log D=\{(\log b,\log d)\mid(b,d)\in D\}.\]

\subsection{Filtrations for point clouds}
\label{sec:geometric-filtrations}

\begin{figure}
	\vertical{\quad\:$\offsetf{}(X)$}
	\includegraphics[width=.1335\linewidth]{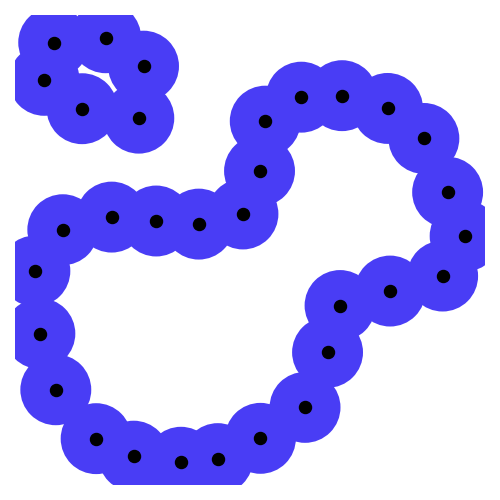}
	\includegraphics[width=.1335\linewidth]{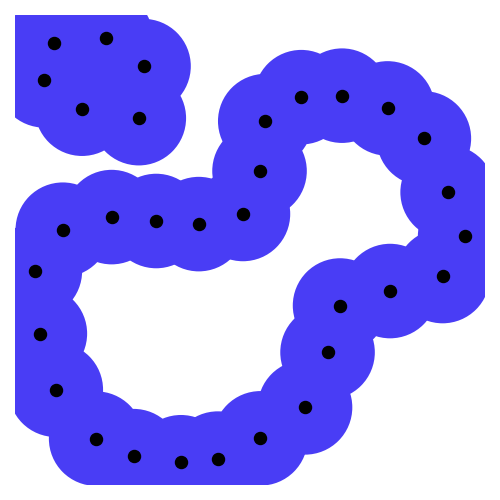}
	\includegraphics[width=.1335\linewidth]{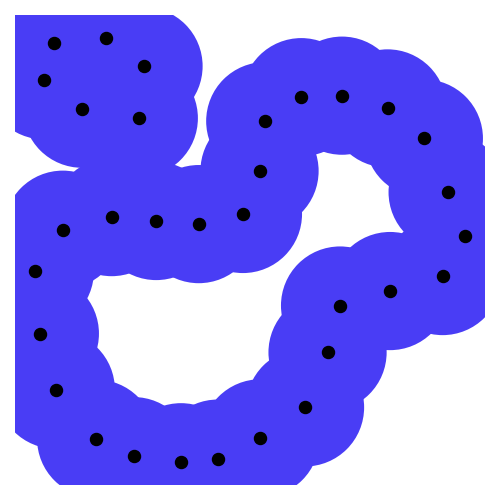}
	\includegraphics[width=.1335\linewidth]{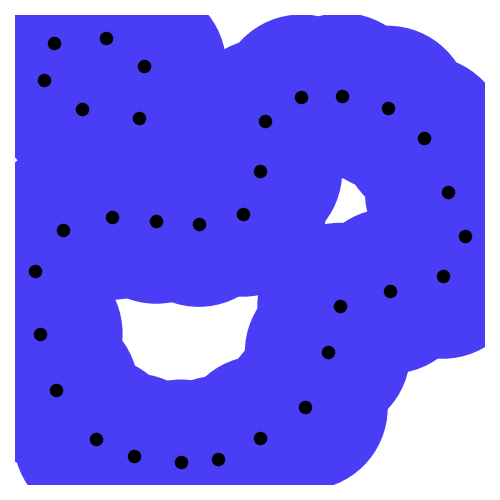}
	\includegraphics[width=.1335\linewidth]{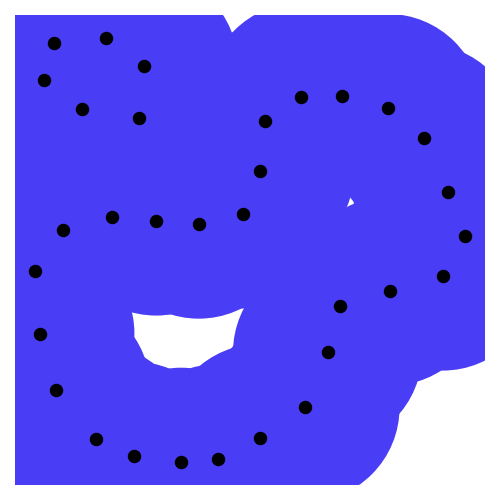}
	\includegraphics[width=.1335\linewidth]{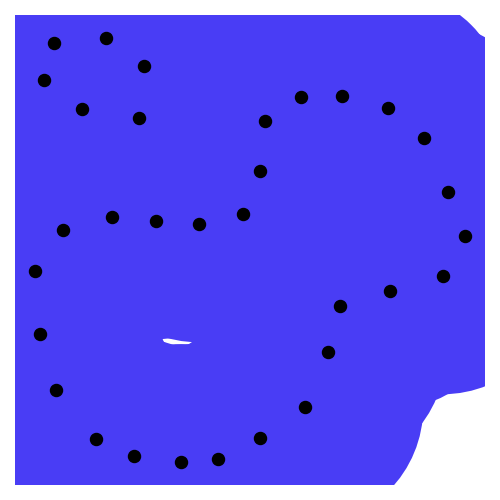}
	\includegraphics[width=.1335\linewidth]{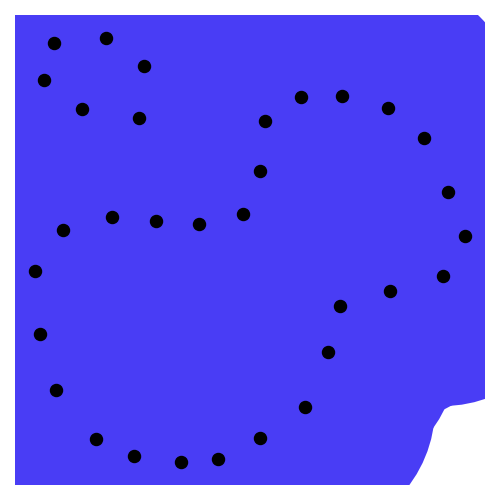} \\\vertical{\quad\:$\ripsf{}(X)$}
	\includegraphics[width=.1335\linewidth]{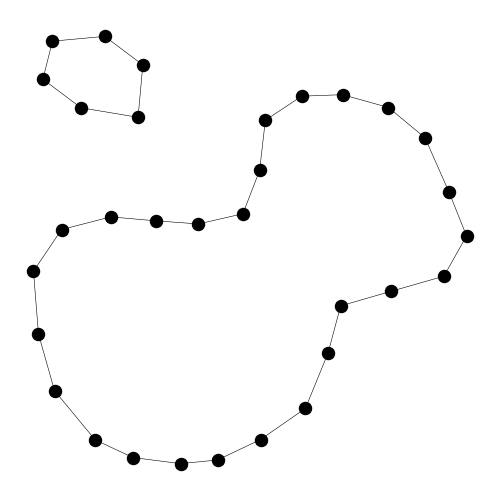}
	\includegraphics[width=.1335\linewidth]{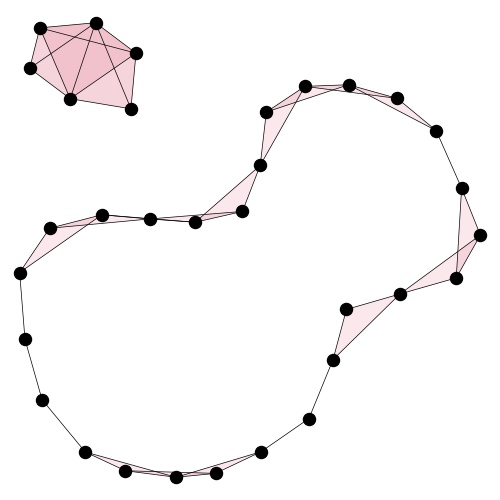}
	\includegraphics[width=.1335\linewidth]{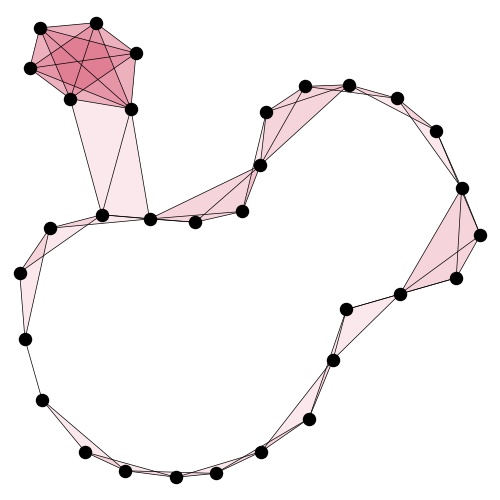}
	\includegraphics[width=.1335\linewidth]{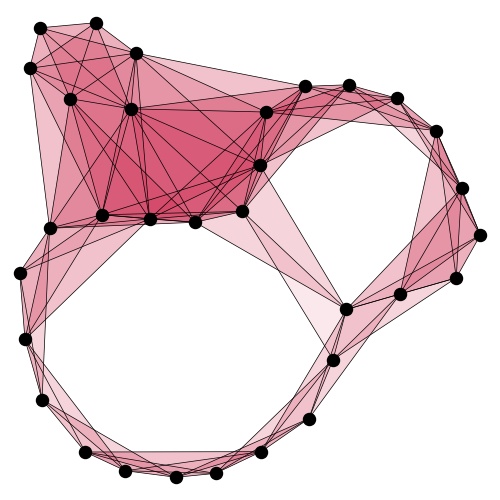}
	\includegraphics[width=.1335\linewidth]{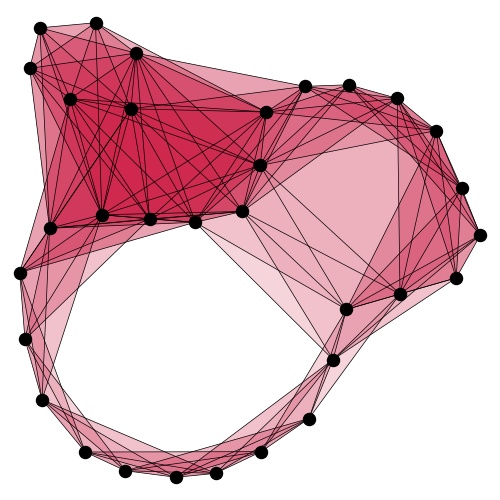}
	\includegraphics[width=.1335\linewidth]{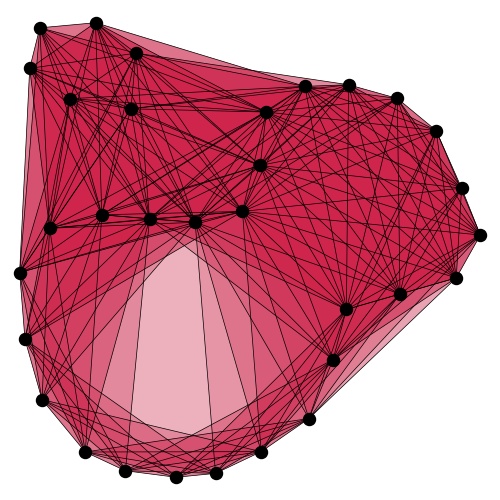}
	\includegraphics[width=.1335\linewidth]{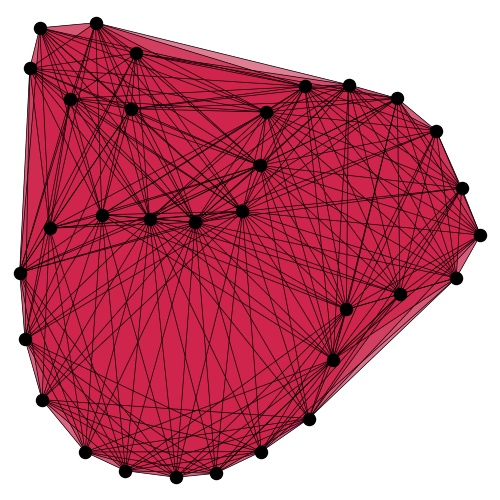} \\\vertical{\quad\:$\delripsf{}(X)$}
	\includegraphics[width=.1335\linewidth]{Figures/Filtrations/Rips_0.15}
	\includegraphics[width=.1335\linewidth]{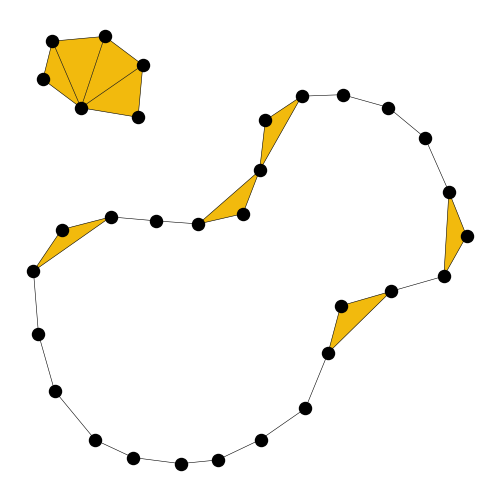}
	\includegraphics[width=.1335\linewidth]{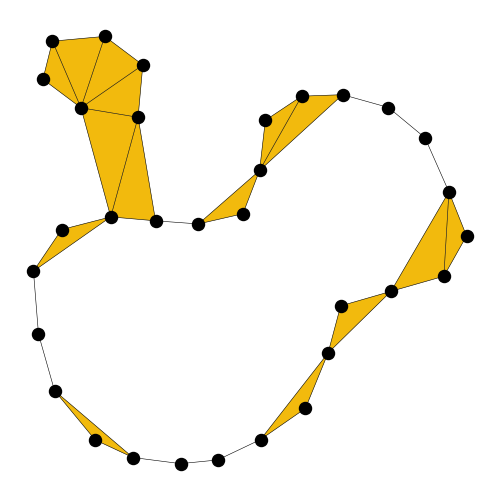}
	\includegraphics[width=.1335\linewidth]{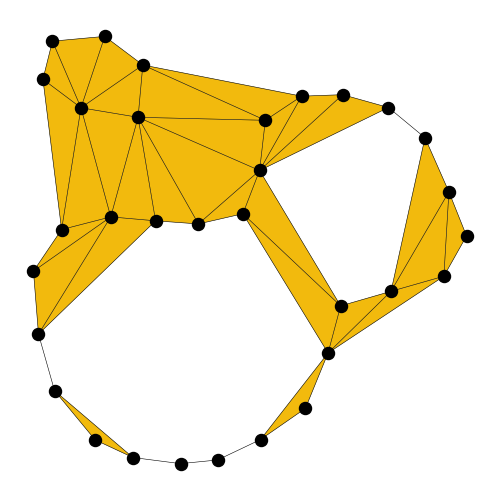}
	\includegraphics[width=.1335\linewidth]{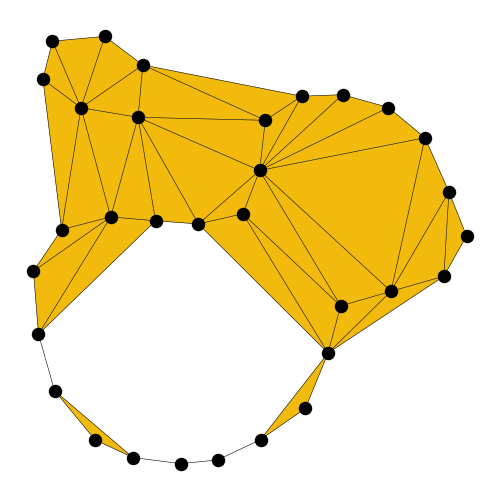}
	\includegraphics[width=.1335\linewidth]{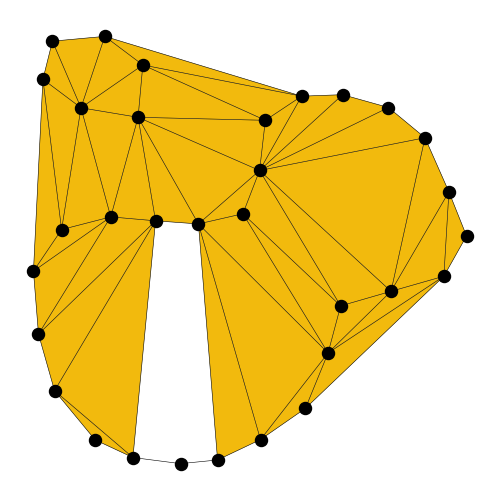}
	\includegraphics[width=.1335\linewidth]{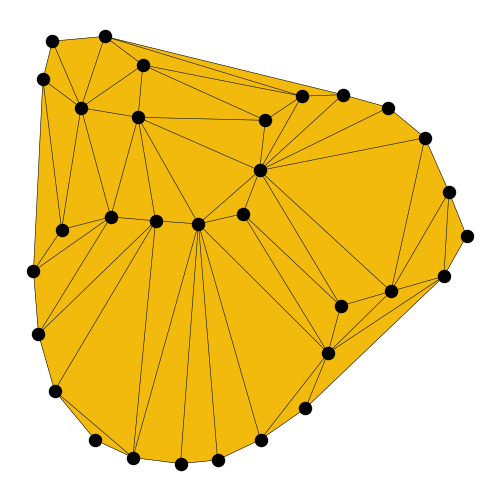}\\
	$r=0.15$ \hspace{.83cm}
	$r=0.20$ \hspace{.83cm}
	$r=0.25$ \hspace{.83cm}
	$r=0.35$ \hspace{.83cm}
	$r=0.40$ \hspace{.83cm}
	$r=0.50$ \hspace{.83cm}
	$r=0.55$
	\centering
	\caption{Offset filtration (top), Rips filtration (middle) and Delaunay--Rips filtration (bottom) for the same point cloud in $\bbr^2$. Notice that the Delaunay--Rips complex has a different number of topological handles than the offset for $r=0.4$, and than the Rips complex for $r=0.5$.}
	\label{fig:filtrations}
\end{figure}

The \emph{offset filtration} $(\offsetf{r}(\pc))_{r\geq0}$, where
$\offsetf{r}(\pc) = \bigcup_{x\in\pc}B_r(x)$ is the union of the closed balls of
radius $r$ centered at the vertices in $\pc$, provides a multi-scale
representation of $\pc$ which is of theoretical interest in topological data
analysis (see \cref{fig:filtrations}, top). However, the offset filtration is hard to use
directly because of its non-combinatorial nature. One may rather use the
\emph{nerves} of this set of closed balls (hereafter denoted the $\nrv$ operator), known as
the \emph{\cech\ filtration} $\bigl(\cechf{r}(\pc)\bigr)_{r\geq0}$, where
\begin{equation}
	\cechf{r}(\pc) =
\left\{\sigma\subset\pc\mid\bigcap_{x\in\sigma}B_r(x)\neq\varnothing\right\} =
\nrv\Bigl(\bigl(B_r(x)\bigr)_{x\in\pc}\Bigr).
	\label{eq:cech}
\end{equation}
By the nerve theorem~\cite{bauer2023unified}, the two filtrations are topologically equivalent, i.e., for each $r\geq0$, $\offsetf{r}(\pc)$ and $\cechf{r}(\pc)$ are homotopy equivalent and therefore have isomorphic homology groups. The scale at which a simplex $\sigma\subset\pc$ is introduced in the \cech\ filtration is the radius of the minimum enclosing ball of $\sigma$, denoted $\meb(\sigma)$, so that we can also write $\cechf{r}(\pc) = \left\{\sigma\subset\pc\mid\meb(\sigma)\leq r\right\}$. A common approximation consists in replacing this value by the diameter $\delta(\sigma)=\max_{x,y\in\sigma}\lVert x-y\rVert_2$. This leads to the well-studied \emph{Rips filtration} $(\ripsf{r}(\pc))_{r\geq0}$, where:\footnote{Rips complexes are often formulated as $\ripsf{r}(\pc) = \left\{\sigma\subset\pc\mid\delta(\sigma)\leq r\right\}$ (i.e., diameter bounded by $r$), but we double here the threshold value to match the scale of \cech{} filtrations (i.e., radius bounded by $r$).}
\begin{equation}
	\ripsf{r}(\pc) = \left\{\sigma\subset\pc\mid\delta(\sigma)\leq2r\right\}.
	\label{eq:rips}
\end{equation}
Jung's theorem~\cite{jung1901ueber} establishes that for any simplex $\sigma$, $\delta(\sigma)\leq2\meb(\sigma)\leq\sqrt{2}\delta(\sigma)$. Therefore, we always have the inclusions $\cechf{r}(\pc)\subseteq\ripsf{r}(\pc)\subseteq\cechf{\sqrt{2}r}(\pc)$ which make the persistence modules of the Rips and \cech\ filtrations multiplicatively strongly $\sqrt{2}$-interleaved (using as $\Phi_r$ and $\Psi_r$ the homomorphisms induced by the inclusion maps $\cechf{r}(\pc)\hookrightarrow \ripsf{r}(\pc)$ and $\ripsf{r}(\pc)\hookrightarrow \cechf{\sqrt2r}(\pc)$). Hence, in the general case, with the strong stability theorem, $\bottle\left(\log\dgmrips{}(\pc), \log\dgmcech{}(\pc)\right)\leq\log\sqrt{2}$.

\paragraph{Delaunay-restricted filtrations}
Both \cech\ and Rips filtrations consider in total $2^n-1$ simplices, or
$\calo\left(n^{\hdim+1}\right)$ simplices if the dimension of those simplices
is bounded up to $\hdim$ (e.g., if we are only interested in the first
$(\hdim-1)$-dimensional persistent homologies). For low-dimensional data, it
is possible to leverage the Delaunay complex in order to reduce the number of
considered simplices and therefore make the computations far more scalable, for
a fixed dimension. In particular, one can consider the intersection of the \cech\
complex with the Delaunay complex of $\pc$, known as the \emph{Delaunay--\cech\
filtration} $\bigl(\delcechf{r}(\pc)\bigr)_{r\geq0}$:
\begin{equation}
	\delcechf{r}(\pc) =
	\left\{\sigma\in\del(\pc)\mid\bigcap_{x\in\sigma}B_r(x)\neq\varnothing\right\}=\nrv\Bigl(\bigl(B_r(x)\bigr)_{x\in\pc}\Bigr)\cap\del(\pc).
	\label{eq:delcech}
\end{equation}
This one is slightly distinct from the more common \emph{Delaunay filtration}, also called the \emph{$\alpha$\nobreakdash-filtration}.\footnote{defined as the nerve of the Voronoi balls, i.e., $\alphaf{r}(\pc) =\left\{\sigma\subset\pc\mid\bigcap_{x\in\sigma}\vor_r(x,\pc)\neq\varnothing\right\}$ where $\vor_r(x,\pc)=B_r(x)\cap\vor(x,\pc)$. The scale at which a simplex appears in this filtration is the minimum radius of an empty circumsphere of $\sigma$, which can be greater than $\meb(\sigma)$.}
However, the \cech, Delaunay--\cech\ and Delaunay filtrations are known to be
topologically equivalent in the sense that they have isomorphic persistent
homology~\cite{bauer2017morse}.
More precisely, for any $r$, there exists a sequence of elementary collapses that transforms $\cechf{r}(\pc)$ into $\delcechf{r}(\pc)$.
Finally, we can also define the
\emph{Delaunay--Rips filtration} $\bigl(\delripsf{r}(\pc)\bigr)_{r\geq0}$
\cite{mishra2023stability}, sometimes referred to as the \emph{weak
$\alpha$\nobreakdash-filtration}~\cite{JMLR:v22:20-325}, as the intersection of
the Rips filtration with the Delaunay complex (see \cref{fig:filtrations},
bottom):
\begin{equation}
    \delripsf{r}(\pc) = \left\{\sigma\in\del(\pc)\mid\delta(\sigma)\leq2r\right\}.
\end{equation}
In addition to featuring fewer simplices because it is restricted to $\del(\pc)$, this is
also faster to compute than the Delaunay--\cech\ filtration, as the computation
of each simplex's minimum enclosing ball is not required. However, Rips and
Delaunay--Rips filtrations are not topologically equivalent (see
\cref{fig:filtrations}).
In addition, the computational cost of the Delaunay triangulation scales exponentially with the ambient dimension $\ddim$ and therefore becomes prohibitively high when $\ddim$ increases. Consequently, the use of this filtration is only conceivable for low-dimensional data (up to the order of 10).

In the following, given a filtration $(\calf_r(\pc))_{r\geq0}$, we denote by $\dgm{\hdim}_{\calf}(\pc)$ the associated $\hdim$-dimensional persistence diagram. We will omit $\pc$ when the considered point cloud is not ambiguous.

\subsection{(In)stability of point cloud persistence}
\label{sec:stability}

Persistence diagrams---for some filtrations---are known to be stable to perturbations of the input. In particular, the bottleneck distance between the persistence diagrams of either Rips or \cech\ filtration of $X$ and $Y$ is upper bounded by their Gromov--Hausdorff distance~\cite{chazal2009gromov, chazal2014persistence}:
\begin{align}
\bottle\bigl(\dgmrips{}(X),\dgmrips{}(Y)\bigr)\leq2\dgh(X,Y)\label{eq:ripsstab}\\
\bottle\bigl(\dgmcech{}(X),\dgmcech{}(Y)\bigr)\leq2\dgh(X,Y)\label{eq:cechstab}
\end{align}
Similar stability results exist for comparing the $L_p$-Wasserstein distances $\calw_p$ and the $p$-norm of the perturbations, when $p<\infty$~\cite{skraba2020wasserstein}.

\begin{figure}
	\centering
	\def\svgwidth{.6\linewidth}
	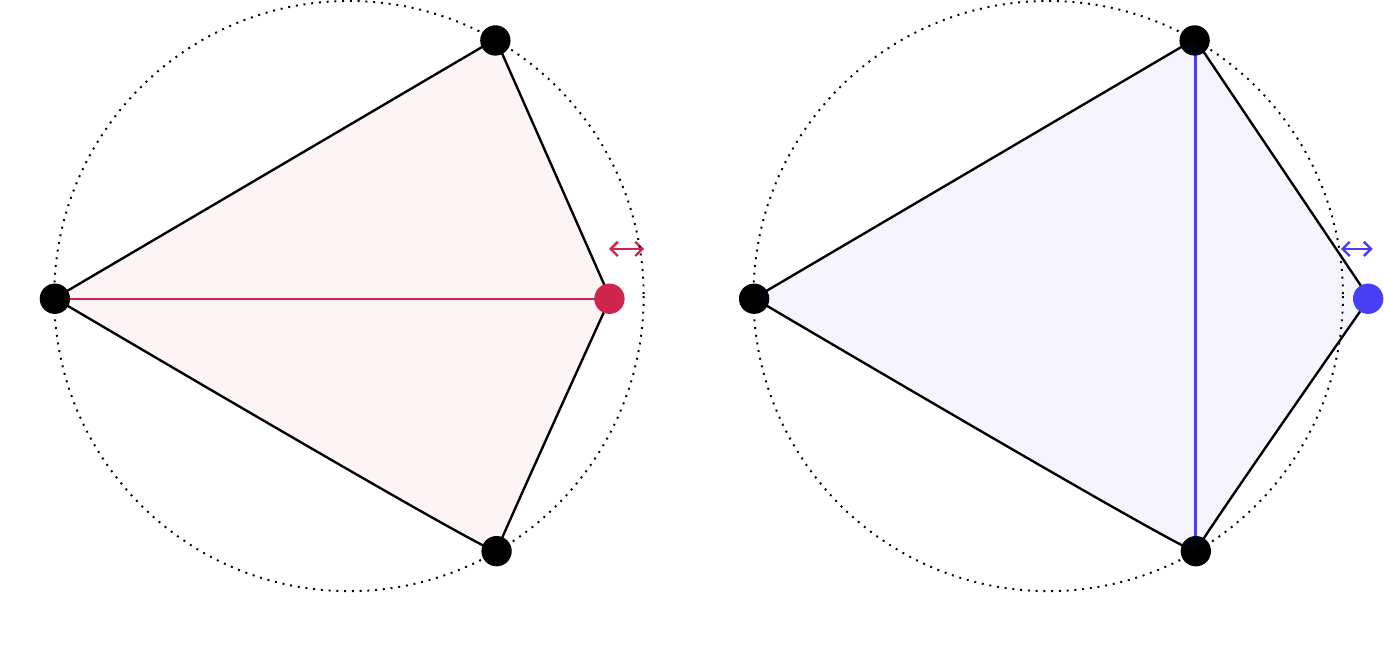
	\vspace{.5cm}
	\caption{Illustration of the instability of persistence diagrams for the Delaunay--Rips filtration. Left: the red point is slightly inside the black unit circle and there exists one positive persistent pair, killed by the red edge. Right: the blue point is slightly outside the unit circle. The blue edge is therefore a Delaunay edge, hence the two triangles appear at the same value $\sqrt{3}/2$ as the black cycle.}
	\label{fig:del-rips-instability}
\end{figure}

On the contrary, persistence diagrams for the Delaunay--Rips filtration may feature instabilities~\cite{mishra2023stability} in the non-general positions for the Delaunay complex.
See \cref{fig:del-rips-instability} for an example of such instability in the plane. Nonetheless, those instabilities can be characterized (see \cref{sec:analysis}).

\subsection{Geometric structures}
\label{sec:geometric-structures}

We introduce below some geometric or combinatorial objects and their links with persistent homology.

\paragraph{Minimum spanning acycles}

\begin{figure}
	\centering
	\hspace{-1cm}
	\includegraphics[width=.2\linewidth]{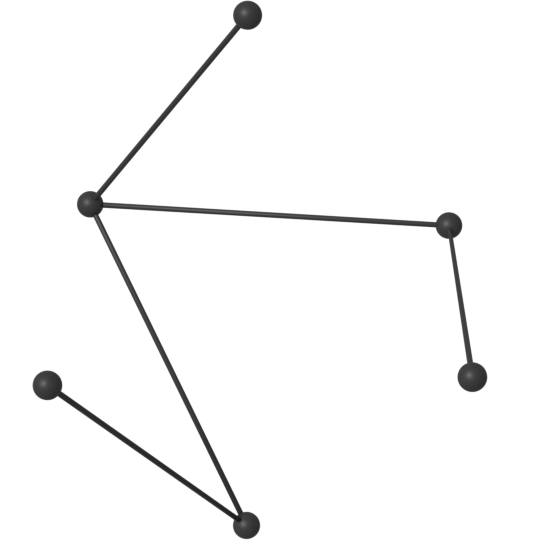}
	\hspace{-.3cm}
	\includegraphics[width=.2\linewidth]{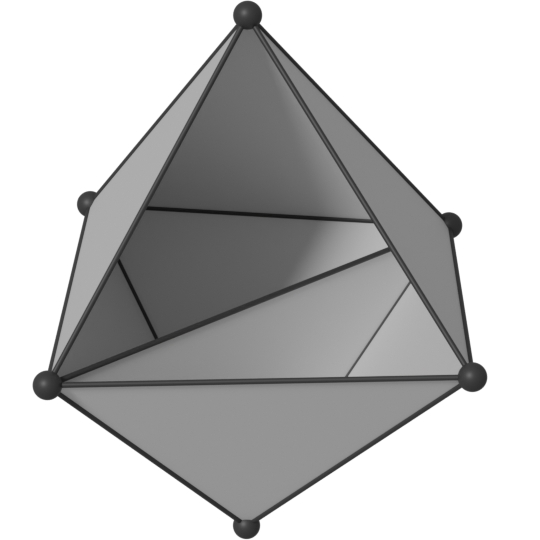}
	\hspace{-.3cm}
	\includegraphics[width=.2\linewidth]{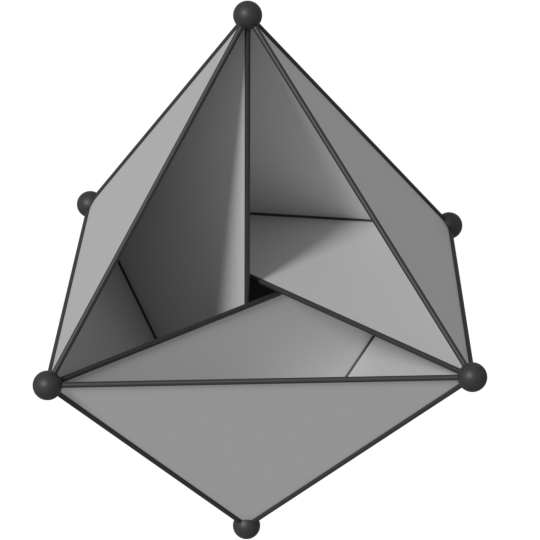}
	\vrule
	\hspace{.5cm}
	\def\svgwidth{.29\linewidth}\footnotesize
\begingroup%
  \makeatletter%
  \providecommand\color[2][]{%
    \errmessage{(Inkscape) Color is used for the text in Inkscape, but the package 'color.sty' is not loaded}%
    \renewcommand\color[2][]{}%
  }%
  \providecommand\transparent[1]{%
    \errmessage{(Inkscape) Transparency is used (non-zero) for the text in Inkscape, but the package 'transparent.sty' is not loaded}%
    \renewcommand\transparent[1]{}%
  }%
  \providecommand\rotatebox[2]{#2}%
  \newcommand*\fsize{\dimexpr\f@size pt\relax}%
  \newcommand*\lineheight[1]{\fontsize{\fsize}{#1\fsize}\selectfont}%
  \ifx\svgwidth\undefined%
    \setlength{\unitlength}{650.75494385bp}%
    \ifx\svgscale\undefined%
      \relax%
    \else%
      \setlength{\unitlength}{\unitlength * \real{\svgscale}}%
    \fi%
  \else%
    \setlength{\unitlength}{\svgwidth}%
  \fi%
  \global\let\svgwidth\undefined%
  \global\let\svgscale\undefined%
  \makeatother%
  \begin{picture}(1,0.68874731)%
    \lineheight{1}%
    \setlength\tabcolsep{0pt}%
    \put(0,0){\includegraphics[width=\unitlength,page=1]{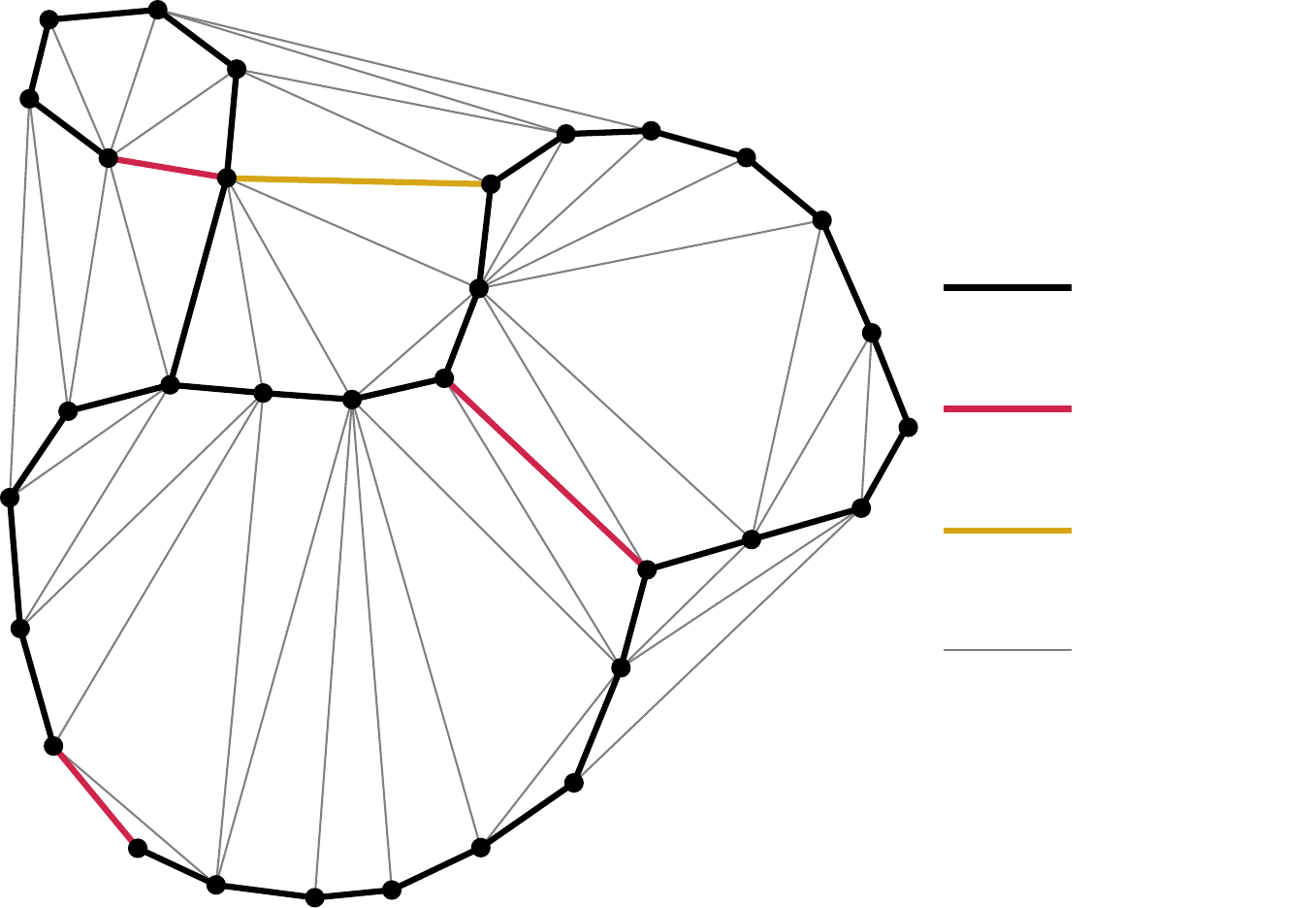}}%
    \put(0.83198638,0.45510829){\color[rgb]{0,0,0}\makebox(0,0)[lt]{\lineheight{1.25}\smash{\begin{tabular}[t]{l}$\mst$\end{tabular}}}}%
    \put(0.83198638,0.36412683){\color[rgb]{0,0,0}\makebox(0,0)[lt]{\lineheight{1.25}\smash{\begin{tabular}[t]{l}$\rng\setminus\mst$\end{tabular}}}}%
    \put(0.83198638,0.18028813){\color[rgb]{0,0,0}\makebox(0,0)[lt]{\lineheight{1.25}\smash{\begin{tabular}[t]{l}$\del\setminus\ug$\end{tabular}}}}%
    \put(0.83198638,0.27220746){\color[rgb]{0,0,0}\makebox(0,0)[lt]{\lineheight{1.25}\smash{\begin{tabular}[t]{l}$\ug\setminus\rng$\end{tabular}}}}%
  \end{picture}%
\endgroup%

	\caption{\textbf{Left:} $\mst(\pc)$ (left), $\msa{2}\bigl(\del(\pc)\bigr)$ (center) and $\msa{2}\bigl(\Delta(\pc)\bigr)$ (right) of the same points $\pc\subset\bbr^3$. For the two $\msa{2}$s, there is no remaining 1-cycle that is unfilled by the triangles (i.e., it is 2-spanning), yet no cavity is created (i.e., it is 2-acyclic).
	\textbf{Right:} illustration of the inclusions $\mst(\pc)\subseteq\rng(\pc)\subseteq\ug(\pc)\subseteq\del(\pc)$ on a planar point cloud $\pc\subset\bbr^2$.
	}
	\label{fig:2msa}
\end{figure}

Minimum spanning acycles, originally introduced in~\cite{kalai1983enumeration}, are a proposed generalization of minimum spanning trees.
Let $\calk$ be a simplicial complex and $S\subset\calk^\hdim$ a set of $\hdim$-simplices.
We say that $S$ is a $\hdim$\nobreakdash-\emph{acycle} if, informally, $S$ does not enclose any $\hdim$-dimensional hole.
This corresponds to $\betti{\hdim}^\star(\calk^{(\hdim-1)}\cup S)=0$.
We say that $S$ is $\hdim$-\emph{spanning} if $S$ fills all the
$(\hdim-1)$-dimensional holes formed by all the $(\hdim-1)$-simplices.
This corresponds to $\betti{\hdim-1}^\star(\calk^{(\hdim-1)}\cup S)=0$.\footnote{Here, we use the \emph{reduced} Betti numbers, which are in brief equivalent to the Betti numbers except when $\hdim=0$, where $\betti{0}^\star=\betti{0}-1$. Intuitively, $\betti{0}^\star$ counts the gaps between connected components, rather than their number.}
Finally, $S$ is said to be a $\hdim$-\emph{spanning acycle} if it has both properties.
The definition of a spanning tree corresponds to that of a 1-spanning acycle: a set of edges $E$ such that $\betti{0}^\star(\pc\cup E)=0$ (spanning-ness: each vertex is reached by the set of edges) and $\betti{1}^\star(\pc\cup E)=0$ (acyclicity).

Given a weight function $w:\calk^\hdim\to\bbr$, a minimum $\hdim$-spanning acycle is a $\hdim$-spanning acycle which minimizes the total weight of its simplices. When $w$ is injective, we have a total ordering on the $\hdim$-faces and the minimum $\hdim$-spanning acycle is unique (\cite{skraba2017randomly}, Lemma 3.8). When it is not (e.g., when $w$ is the diameter $\delta$), we only have a partial ordering, that can, however, be extended to a total order. In this case, the minimum $\hdim$-spanning acycle is unique with respect to the chosen total order, but its weights are independent of this choice (\cite{skraba2017randomly}, Lemma 3.14).
When $\hdim=1$, it corresponds to the definition of a minimum spanning tree: in particular, with the diameter as weight, we have $\mst(\pc)=\msa{1}\big(\Delta(\pc)\big)=\msa{1}\bigl(\del(\pc)\bigr)$. \cref{fig:2msa}, left, depicts minimum 1- and 2-spanning acycles for $\pc\subset\bbr^3$.

Minimum spanning acycles have only recently received significant attention, as they have interesting properties with birth and death values in persistence diagrams, as noticed in~\cite{skraba2017randomly, skraba2025central}. In particular, let $S\subset\calk^\hdim$ be a $\hdim$-minimum spanning acycle for a weight function $w$. Then $\{w(\sigma)\mid\sigma\in S\}$ is the set of death times of the $\PH{\hdim-1}$ classes of the filtration induced by $w$ on $\calk$, while $\{w(\sigma)\mid\sigma\in\calk^\hdim\setminus S\}$ is the set of birth times of the $\PH{\hdim}$ classes.
This is a generalization of the fact that the edges of the minimum spanning tree $\mst(\pc)$ are exactly the edges killing the \PH{0} classes of the (Delaunay)--Rips filtration of $\pc$ (and that the remaining edges give birth to \PH{1} classes).
This is commonly used to speed up the computation of 0-dimensional persistence diagrams thanks to a union-find data structure~\cite{cormen}.

\paragraph{Neighborhood graphs}

The relative neighborhood graph~\cite{toussaint_relative_1980} of $\pc$ is defined as the set of edges $e=(x,y)$ such that there is no third point that is closer to $x$ and $y$ than they are to each other, i.e., $\rng(\pc)=\{(x,y)\in\pc^2\mid\not\exists z\in\pc,\dist{x}{z}<\dist{x}{y}\text{ and }\dist{y}{z}<\dist{x}{y}\}$.
It is a subgraph of the Delaunay complex which can be computed in time complexity $\calo(n\log n)$ in the plane~\cite{supowit_relative_1983, jaromczyk_note_1987, lingas_linear-time_1994}, but not in higher dimensions.
Another graph of interest is the Urquhart graph $\ug(\pc)$ \cite{urquhart_algorithms_1980} which is defined as the subgraph of the Delaunay complex where the longest edge of each Delaunay triangle has been removed.
Under the general position hypothesis, we have the inclusions $\mst(\pc)\subseteq\rng(\pc)\subseteq\ug(\pc)\subseteq\del(\pc)$ (see \cref{fig:2msa}, right).
It was shown in~\cite{koyama2023faster} that the edges of $\rng(\pc)\setminus\mst(\pc)$ are exactly the edges that create a \PH{1} pair of positive persistence for the Rips filtration of $\pc$.

\section{Analysis}
\label{sec:analysis}

In this section, we establish how Delaunay--Rips persistence diagrams approximate Rips persistence diagrams, and which instabilities can happen when the input is perturbed.
Proofs are in \cref{appendix:proofs}.

Jung's theorem~\cite{jung1901ueber} states that any compact set $K\subset\bbr^\ddim$ of diameter $\delta(K)=\max_{x,y\in K}\lVert x-y\rVert_2$ is enclosed by a closed ball of radius $\meb(K)$, with:
\begin{align}
	\label{eq:jung}
	\delta(K) \leq 2\meb(K) \leq \delta(K)\underbrace{\sqrt{\frac{2\ddim}{\ddim+1}}}_{\triangleq\jung{\ddim}}
\end{align}
where the right inequality is precisely attained by the regular $\ddim$-simplices. In addition, for any simplicial complex $\calk$,
any $\hdim$-simplex $\sigma\in\calk^\hdim$ lives in a $\hdim$-dimensional affine space, hence:
\begin{equation}
	\label{eq:jung_simplex}
	\delta(\sigma) \leq 2\meb(\sigma) \leq \jung{\hdim} \delta(\sigma).
\end{equation}
Therefore, as the \cech\ and Delaunay--\cech\ filtrations write respectively as
$\cechf{r}(\pc) = \left\{\sigma\subset\pc\mid\meb(\sigma)\leq r\right\}$ and
$\delcechf{r}(\pc) = \left\{\sigma\in\del(\pc)\mid\meb(\sigma)\leq r\right\}$,
we have the following inclusions between their $\hdim$-dimensional skeletons:
\begin{equation}
	\label{eq:inclusions_jung}
	\cechf{r}^{(\hdim)}(\pc) \subseteq
	\ripsf{r}^{(\hdim)}(\pc) \subseteq
	\cechf{\jung{\hdim}r}^{(\hdim)}(\pc)
	\quad\text{and}\quad
	\delcechf{r}^{(\hdim)}(\pc) \subseteq
	\delripsf{r}^{(\hdim)}(\pc) \subseteq
	\delcechf{\jung{\hdim}r}^{(\hdim)}(\pc).
\end{equation}
In addition, due to $\del(\pc)\subseteq\Delta(\pc)$, we also have the following inclusions:
\begin{equation}
	\label{eq:inclusions_del}
	\delcechf{r}^{(\hdim)}(\pc) \subseteq \cechf{r}^{(\hdim)}(\pc)
	\quad\text{and}\quad
	\delripsf{r}^{(\hdim)}(\pc) \subseteq \ripsf{r}^{(\hdim)}(\pc).
\end{equation}
Applying the homology functor to these simplicial inclusions between the ($\hdim+1$)-skeletons induces morphisms between the corresponding $\hdim$-dimensional homology groups, for example, for \cref{eq:inclusions_jung}, 
$\hgf{\hdim}{\cechn}{r}\longrightarrow\hgf{\hdim}{\ripsn}{r}\longrightarrow\hgf{\hdim}{\cechn}{\jung{\hdim+1}r}$.
In addition, $\cechf{r}$ can be transformed into $\delcechf{r}$ with a sequence of elementary collapses \cite{bauer2017morse}.
On the homology level, this not only implies that $\hgf{\hdim}{\delcechn}{r}$ and $\hgf{\hdim}{\cechn}{r}$ are isomorphic, but also (see, e.g., \cite{hatcher2005algebraic}) that for all $r$, the simplicial inclusion
\begin{tikzcd}[cramped]
	\iota_r:\delcechf{r}^{(\hdim)}\arrow[r,hook]&\cechf{r}^{(\hdim)}
\end{tikzcd}
induces an isomorphism:
\begin{equation}
	\label{eq:induced_isomorphism}
	\begin{tikzcd}[cramped, sep=scriptsize]
		\iota_r^\star=H_\hdim(\iota_r):\hgf{\hdim}{\delcechn}{r} \arrow[r,"\simeq"] & \hgf{\hdim}{\cechn}{r}
	\end{tikzcd}.
\end{equation}

\subsection{Rips and Delaunay--Rips comparison}
\label{sec:comp_r_dr}

Thanks to these relations, we can now show that the Delaunay--Rips persistence module provides an approximation of the Rips persistence module.

\begin{theorem}
	\label{lemma:rips_delrips}
	The $\hdim$-dimensional Rips and Delaunay--Rips persistence modules $\big(\hgfx{\hdim}{\ripsn}{r}{\pc}\big)_{r\geq0}$ and $\big(\hgfx{\hdim}{\delripsn}{r}{\pc}\big)_{r\geq0}$ are multiplicatively strongly $\jung{\hdim+1}$-interleaved, hence the following bound, which is tight:
	\begin{equation}
		\label{eq:bottle_logdgms}
		\bottle\left(\log\dgmdr{\hdim}(\pc),\log\dgmrips{\hdim}(\pc)\right)\leq\log\jung{\hdim+1}.
	\end{equation}
\end{theorem}

For the case of $\hdim=0$, \cref{eq:bottle_logdgms} gives more directly that the 0-th Rips and Delaunay--Rips persistence modules are isomorphic, and thus that $\dgmrips{0}(\pc)=\dgmdr{0}(\pc)$. In the other cases ($\hdim\geq1$), this only gives an upper bound on the bottleneck distance between logarithmic diagrams. In practice, one could be more interested in a bound between the persistence diagrams themselves, which has to scale with the diameter of $\pc$:
\begin{align}
	\label{eq:bottle_dgms}
	\bottle\left(\dgmdr{\hdim}(\pc),\dgmrips{\hdim}(\pc)\right)\leq(\jung{\hdim+1}-1)\delta(\pc).
\end{align}

\subsection{Delaunay--Rips instability}
\label{sec:dr_instability}


The following result quantifies the possible instabilities of Delaunay--Rips persistence diagrams that may happen near a Delaunay non-general position (see \cref{fig:del-rips-instability} for an illustration):\footnote{while combining \cref{eq:ripsstab,eq:bottle_dgms} with the triangle inequality would give $\bottle\left(\dgmdr{\hdim}(X),\dgmdr{\hdim}(Y)\right)\leq\left(\jung{\hdim+1}-1\right)\big(\delta(X)+\delta(Y)\big)+ 2\dgh(X,Y)$.}
\begin{theorem}
	\label{lemma:delrips_instab}
	Let $X$ and $Y$ be two point clouds. Then for any $\hdim\geq0$:
	\begin{equation}
		\bottle\left(\dgmdr{\hdim}(X),\dgmdr{\hdim}(Y)\right)
		\leq\left(\jung{\hdim+1}-1\right)\max\{\delta(X),\delta(Y)\} + 2\dgh(X,Y).
		\label{eq:bottle_stab2}
	\end{equation}
\end{theorem}

We empirically evaluate in \cref{sec:distance-stability-practice} how well the Rips module is approximated (\cref{lemma:rips_delrips}) and the stability of Delaunay--Rips persistence diagrams (\cref{lemma:delrips_instab}) for several point set distributions in $\bbr^3$.

\section{Algorithm}
\label{sec:method}

In this section, we introduce an algorithm specifically dedicated to the computation of Delaunay--Rips persistence diagrams.
We begin with a few definitions and lemmas (\cref{subsec:generalities}), before presenting the algorithm (\cref{subsec:description}) and how it can be implemented (\cref{subsec:details,subsec:parallelization}).
Proofs are in \cref{appendix:proofs}.

\subsection{Generalities}
\label{subsec:generalities}

Let $\pc\subset\bbr^d$ and $\calk$ be either $\del(\pc)$ or $\Delta(\pc)$.
For each $\hdim\geq1$, we build a total order $\order{\hdim}$ on 
$\calk^\hdim$ that extends the partial order induced by the diameter, i.e., 
for all $\sigma,\sigma'\in\calk^\hdim$, if $\delta(\sigma)<\delta(\sigma')$, 
then $\sigma\order{\hdim}\sigma'$.
More precisely, we define by induction on $\hdim\geq1$:
\begin{align}
\begin{split}
	\label{eq:def_order}
	\sigma\order{1}\sigma'&\iff
		\delta(\sigma)<\delta(\sigma')
		\text{ or }
		(\delta(\sigma)=\delta(\sigma') \text{ and } \sigma\lex\sigma')\\
	\sigma\order{\hdim+1}\sigma'&\iff
		\max\partial\sigma\order{\hdim}\max\partial\sigma'
		\text{ or }
		(\max\partial\sigma=\max\partial\sigma' \text{ and } \sigma\lex\sigma')
\end{split}
\end{align}
where $\max$ is understood as with respect to $\order{\hdim}$ and where $\lex$ is the lexicographic order between simplices sorted with increasing vertices.
In the following, we denote more concisely $\prec$ instead of $\order{\hdim}$ as $\hdim$ is given by $\sigma$ and $\sigma'$.

As mentioned in \cref{sec:geometric-structures}, the minimum $\hdim$-spanning acycle is unique with respect to the chosen total order~\cite{skraba2017randomly} (even though the diameters of the simplices it contains are independent of this choice), so that we write in the following $\msa{\hdim}(\pc)$.
We omit specifying $\pc$ when this is unambiguous.


\paragraph{Relative neighborhood simplices}
We define below the $\hdim$-dimensional relative neighborhood simplices $\rnsc{\hdim}$, for $1\leq\hdim<\ddim$:
\begin{equation}
    \rnsc{\hdim}(\calk) = \left\{\sigma\in\calk^\hdim\mid\forall\tau\in\cof(\sigma, \calk),\sigma\prec\max\partial\tau\right\}
    \label{eq:rnsc}
\end{equation}
where $\cof(\sigma,\calk) = \{\tau\in\calk\mid\sigma\subset\tau\text{ and }\dim\tau=\dim\sigma+1\}$ is the set of cofacets of a simplex $\sigma$ in $\calk$.
This construction offers a generalization---in terms of dimension $\hdim$---of both Urquhart and relative neighborhood graphs, as
$\ug(\pc)=\rnsc{1}\big(\del(\pc)\big)$ and
$\rng(\pc)=\rnsc{1}\big(\Delta(\pc)\big)$.
As we focus on the Delaunay--Rips
filtration, we now set $\calk=\del(\pc)$ and we denote more concisely the
\emph{$\hdim$-dimensional Urquhart simplices}
$\usc{\hdim}(\pc)=\rnsc{\hdim}\big(\del(\pc)\big)$ (see example in
\cref{fig:usc}).

\paragraph{Apparent pairs}
An apparent pair~\cite{bauer_ripser_2021} is a pair $(\sigma,\tau)$ such that $\sigma$ is the youngest facet of $\tau$ and $\tau$ is the oldest cofacet of $\sigma$.
Note that any apparent pair is indeed a persistent pair (see Lemma 3.3 of~\cite{bauer_ripser_2021}).
In addition, for filtrations based on the diameter (e.g., Rips or Delaunay--Rips filtrations), any $\hdim$-dimensional apparent pair $(\sigma,\tau)$, with $\hdim\geq1$, has persistence zero, since in this case  $\delta(\tau)=\max_{\sigma'\in\partial\tau}\delta(\sigma')=\delta(\sigma)$.
In the following we call \emph{zero} (resp. \emph{positive}) \emph{persistent pairs} the persistent pairs of persistence zero (resp. of positive persistence).
We show below some links that the above construction has with
(apparent) zero and positive persistence pairs.
Before that, we first cite a lemma that will be useful in the following.

\begin{lemma}[Exchange property, see Lemma 3.7 of~\cite{skraba2017randomly}]
	\label{lemma:exchange}
	Let $S\subset\calk^\hdim$ be a $\hdim$-spanning acycle of $\calk$ and $\sigma\in\calk^\hdim\setminus S$. For any $\sigma'\in S$ such that $\sigma'$ is part of a $\hdim$-cycle containing $\sigma$, $S\cup\{\sigma\}\setminus\{\sigma'\}$ is also a $\hdim$-spanning acycle of $\calk$.
\end{lemma}

The following result generalizes the fact that in general position, $\mst$ is included in both $\rng$ and $\ug$ (see \cref{sec:geometric-structures}).
\begin{lemma}
	\label{lemma:msa_in_usc}
	For all $1\leq\hdim<\ddim$, $\msa{\hdim}\subset\usc{\hdim}$.
\end{lemma}

\begin{lemma}
	\label{lemma:non_apparent_in_usc}
	All $\hdim$-simplices that give birth to a positive $\hdim$-dimensional persistent pair are in $\usc{\hdim}\setminus\msa{\hdim}$.
\end{lemma}

\begin{figure}
	\hspace{-.3cm}
	\raisebox{.5cm}{\includegraphics[width=.25\linewidth]{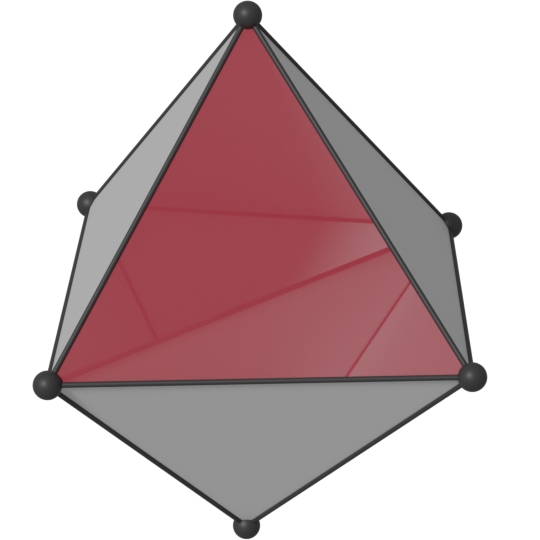}}
	\vrule
	\hspace{3.3cm}
	\def\svgwidth{0.3\linewidth}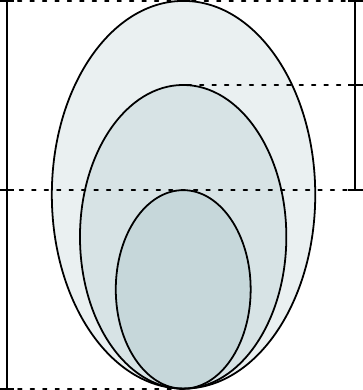
	\caption{
		\textbf{Left:} $\usc{2}(X)$ shown for $\pc\subset\bbr^3$, with $\usc{2}(X)\setminus\msa{2}(X)$ highlighted in transparent red.
		\textbf{Right:} Summary of the inclusions between the different sets of simplices and their role.}
	\label{fig:usc}
\end{figure}

We can establish a more precise result for \PH{1}.
Indeed, it is shown in~\cite{koyama2023faster} that, for the Rips filtration, the edges in $\rng\setminus\mst$ exactly correspond to the births of the positive \PH{1} pairs.
We establish below an analogous result for the Delaunay--Rips filtration, involving $\ug$ instead of $\rng$.
\begin{lemma}[Theorem 3.10 of~\cite{bauer_ripser_2021}]
	\label{lemma:ph1_apparent_is_zero}
	Let $X$ be a point set with distinct pairwise distances.
	Among the \PH{1} pairs, the apparent pairs are exactly the zero persistence pairs.
\end{lemma}
\begin{lemma}
	\label{lemma:ph1_ug_is_positive}
	Let $\pc$ be a point set in general position (unique pairwise distances, in addition to the general position for the Delaunay complex).
	The edges in $\ug\setminus\mst$ exactly correspond to the births of \PH{1} pairs with a positive persistence for the Delaunay--Rips filtration.
\end{lemma}

Finally, the choice made for the total order (\cref{eq:def_order}) enables the following result:
\begin{lemma}
	\label{lemma:usc_in_msa}
	Let $\sigma$ be a $\hdim$-simplex.
	If $\sigma\prec\max\partial\tau$ for all $\tau\in\cof(\sigma,\msa{\hdim+1})$, then $\sigma\in\usc{\hdim}$.
	In other words, to determine whether a simplex is in $\usc{\hdim}$, it suffices to check if it is not the largest facet of any of its cofacets in $\msa{\hdim+1}$ (instead of all its cofacets).
\end{lemma}

\subsection{Description}
\label{subsec:description}

In this section, we turn to the description of an algorithm that computes the persistence diagrams of the Delaunay--Rips filtration of a point cloud in arbitrary dimension. The idea is to group $\hdim$-simplices into $\hdim$-cells delimited by $\usc{\hdim-1}$ in order to reduce the number of reduction steps in the computation of $\PH{\hdim-1}$.

\subsubsection{Codimension-1 persistent homology}
\label{subsubsec:codim1}

For the specific case of persistent homology of codimension 1 (i.e., \PH{\ddim-1}), we can rely on a duality approach.
For that, recall that \cref{lemma:non_apparent_in_usc} means that the simplices that are not in $\usc{\ddim-1}$ correspond to births of apparent, zero \PH{\ddim-1} pairs.
Therefore, to determine the non-apparent \PH{\ddim-1} pairs, we consider the $\ddim$-cells delimited by $\usc{\ddim-1}$---in other words, we group the $\ddim$-simplices by connected components in the complement of $\usc{\ddim-1}$. They geometrically correspond to $\ddim$-polytopes. We can endow each of them with the diameter of the largest Delaunay simplex they contain, as their filtration value.
Then, by duality, finding \PH{\ddim-1} is equivalent to computing \PH{0} of the dual graph where dual nodes are these $\ddim$-cells and dual edges are the $\usc{\ddim-1}$ simplices between them, with all weights set to their opposite value.
All $\usc{\ddim-1}$ simplices that are unpaired at the end of the reduction give $\msa{\ddim-1}$.
Note that this is equivalent to performing a reverse-delete algorithm on $\usc{\ddim-1}$ to obtain $\msa{\ddim-1}$ (like minimum spanning trees, minimum spanning acycles can be computed thanks to greedy algorithms like Kruskal's or reverse-delete~\cite{skraba2017randomly}).

\subsubsection{Intermediate persistent homologies}
\label{subsubsec:intermediate}

For computing \PH{\hdim}, with $1\leq\hdim\leq\ddim-2$, we follow the same strategy to ignore zero persistent pairs, i.e., by gathering $(\hdim+1)$-simplices into $(\hdim+1)$-cells by leveraging \cref{lemma:non_apparent_in_usc}.
The first step is to compute $\usc{\hdim}$ that can be obtained from $\msa{k+1}$ thanks to \cref{lemma:usc_in_msa}. However, unlike the codimension-1 case, $\usc{\hdim}$ does not unambiguously define $(\hdim+1)$-cells within $\msa{\hdim+1}$, since the latter has a non-manifold structure. Indeed, it contains $\hdim$-simplices that are connected to more than two $(\hdim+1)$-simplices (see \cref{fig:3dmethod-concepts} for $\ddim=3$ and $\hdim=1$). We define those \emph{non-manifold} $\hdim$-simplices as:
\begin{equation}
	\label{eq:nmk}
	\nmj{\hdim} = \left\{\sigma\in\del^{\hdim}\mid|\cof(\sigma,\msa{\hdim+1})|>2\right\}.
\end{equation}
Then, we instead consider the $(\hdim+1)$-cells that are---this time unambiguously---delimited within $\msa{\hdim+1}$ by the union $\usc{\hdim}\cup\nmj{\hdim}$, at the price of dealing with some additional apparent pairs (exactly $|\nmj{\hdim}\setminus\usc{\hdim}|$).
These $(\hdim+1)$-cells are given as filtration value the diameter of the largest simplex they contain.
To obtain the $\PH{\hdim}$ persistence pairs, we run a boundary matrix reduction with these $(\hdim+1)$-cells as columns, and with the simplices in $\usc{\hdim}\cup\nmj{\hdim}$ as rows.
The $\hdim$-simplices that are still unpaired at the end of the reduction exactly give $\msa{\hdim}$. \cref{algo:intermediary-persistence} gives an overview of the algorithm and \cref{subsec:details} describes more precisely each step and the data structures they use.
\begin{figure}
	\begin{center}
		\includegraphics[height=.27\linewidth]{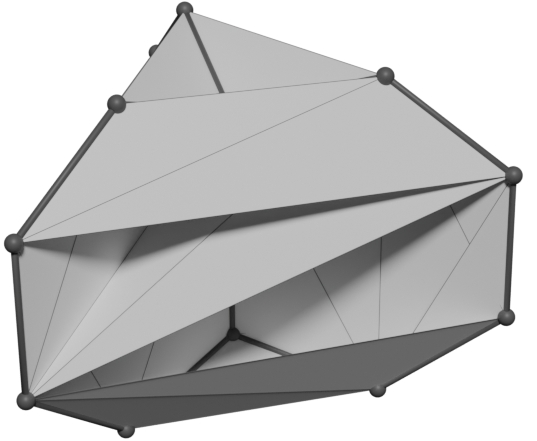}\hspace{2mm}
		\includegraphics[height=.27\linewidth]{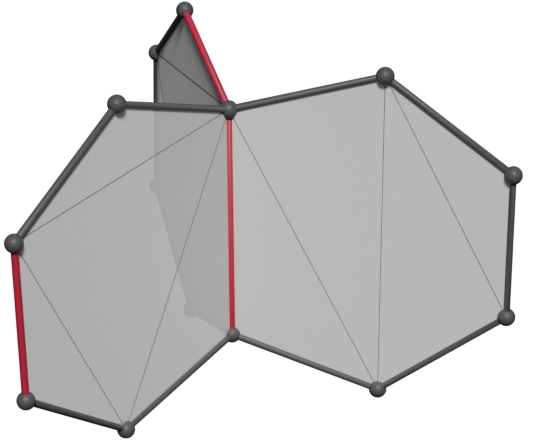}\hspace{2mm}
		\includegraphics[height=.27\linewidth]{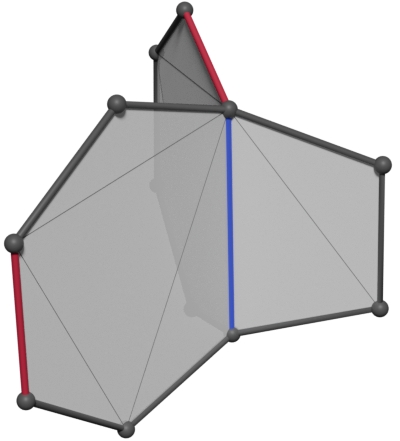}
	\end{center}
	\caption{
	\textbf{Left:} A point cloud $\pc\subset\bbr^3$ and its $\msa{2}$, with $\ug$ highlighted.
	\textbf{Center:} This point cloud has three \PH{1} positive persistent pairs. The three edges in $\ug\setminus\mst$ are in red (one of them is also in $\nmj{1}$). The edges of the $\mst$ are the black and thick edges.
	\textbf{Right:} Same thing, but with only two \PH{1} positive persistent pairs. However, here, there is an edge in $\nmj{1}\setminus\ug$, shown in blue. Thus, three finite polygons are still considered in $\msa{2}$, and one of them will give an apparent pair.
	}
	\label{fig:3dmethod-concepts}
\end{figure}

\begin{algorithm}
	\DontPrintSemicolon
	\caption{Obtain $\PH{\hdim}$ from $\PH{\hdim+1}$}
	\label{algo:intermediary-persistence}
	\KwIn{Point cloud $\pc\subset\bbr^\ddim$, $\msa{\hdim+1}$}
	Compute $\usc{\hdim}$\;
	Identify non-manifold simplices $\nmj{\hdim}$ in $\msa{\hdim+1}$\;
	Determine the $(\hdim+1)$-cells formed by $\usc{\hdim}\cup\nmj{\hdim}$ within $\msa{\hdim+1}$\;
	Run boundary matrix reduction (columns are the $(\hdim+1)$-cells; rows are $\usc{\hdim}\cup\nmj{\hdim}$) to obtain $\dgmdr{\hdim}$\;
	Unpaired $\hdim$-simplices give $\msa{\hdim}$\;
	\KwOut{$\dgmdr{\hdim}$ and $\msa{\hdim}$}
\end{algorithm}

\subsubsection{1-dimensional persistent homology}
\label{subsubsec:dim1}


\begin{algorithm}
	\DontPrintSemicolon
	\caption{Obtain \PH{0} and \PH{1}}
	\label{algo:1-persistence}
	\KwIn{Point cloud $\pc\subset\bbr^\ddim$, $\msa{2}$}
	Extract $\ug$ and $\nmj{1}$ from $\msa{2}$ (see \cref{lemma:usc_in_msa})\;
	Apply Kruskal's algorithm on $\ug$ to compute $\mst$ and $\dgmdr{0}$\;
	Determine the 2-cells formed by $\ug\cup\nmj{1}$ in $\msa{2}$\;
	Run boundary matrix reduction (columns are the 2-cells; rows are $(\ug\cup\nmj{1})\setminus\mst$) to obtain $\dgmdr{1}$\;
	\KwOut{$\dgmdr{0}$ and $\dgmdr{1}$}
\end{algorithm}

For the specific case of 1-dimensional persistent homology, we can exploit the additional knowledge of the edges that give birth to \PH{1} pairs before the reduction. Indeed, among $\ug\cup\nmj{1}$ edges, only those in $(\ug\cup\nmj{1})\setminus\mst$ create \PH{1} pairs.
Therefore, we can run the boundary matrix reduction with the 2-cells as columns and with the edges in $(\ug\cup\nmj{1})\setminus\mst$ only as rows.
\cref{algo:1-persistence} gives an overview of the algorithm and \cref{subsec:details} describes more precisely each step and the data structures they use.

Note that for the specific context of $\pc\subset\bbr^3$, only the codimension-1 case (for \PH{2}, see \cref{subsubsec:codim1}) and this 1-dimensional case (for \PH{1} and \PH{0}) are required.

\subsection{Data structures}
\label{subsec:details}

We describe more precisely each step of the algorithms described in \cref{subsec:description} and their data structures.


\paragraph{Delaunay complex}

We use CGAL~\cite{fabri_cgal_2009} to compute the Delaunay complex of $\pc\subset\bbr^\ddim$, with identifiers on vertices and \emph{full cells}, and a virtual point at infinity to deal with its boundary (i.e., the border of the convex hull of $\pc$). This implementation offers the possibility to construct the Delaunay complex in parallel only for $\ddim=3$.

\paragraph{Urquhart facets and cells}

$\usc{\ddim-1}$ is built by enumerating the Delaunay $(\ddim-1)$-simplices and checking the local condition that they are not the largest (for $\prec$) facet of one of their adjacent $\ddim$-simplices (see \cref{eq:rnsc}).
During this computation, we maintain a union-find data structure on the $\ddim$-simplices that enables deducing efficiently the \emph{Urquhart $\ddim$-cells} delimited by $\usc{\ddim-1}$.
More precisely, when we encounter a non-$\usc{\ddim-1}$ simplex, we merge the two classes associated with its (at most two) cofacets.
At the end of the computation, the Urquhart $\ddim$-cells are implicitly represented by the roots of this union-find data structure.
In addition, for each class, we retrieve the largest $\ddim$-simplex it contains.
 
\paragraph{Codimension-1 persistent pairs}

Computing \PH{\ddim-1} involves pairing these Urquhart $\ddim$-cells with the $\usc{\ddim-1}$ simplices, which can be solved by relying on duality.
Indeed, we can rather compute \PH{0} of the dual graph where dual nodes are the Urquhart cells with the opposite diameter of their largest inner $\ddim$-simplex as their filtration value, and where dual edges are the $\usc{\ddim-1}$ simplices with their opposite diameter as their filtration value.
Simultaneously, we construct $\msa{\ddim-1}$ as the $(\ddim-1)$-simplices whose cofacets are already in the same union-find class (as deleting them would create a $(\ddim-1)$-cycle).

\paragraph{Other persistent pairs}

We can compute $\usc{\hdim}$ by enumerating the $\hdim$-simplices in $\msa{\hdim+1}$ (i.e., all the Delaunay $\hdim$-simplices, since it is spanning).
Leveraging \cref{lemma:usc_in_msa}, it suffices to check the local condition that they are not the largest facet of one of their cofacets in $\msa{\hdim+1}$.
We also store the connectivity from $\hdim$-simplices to $\msa{\hdim+1}$ simplices in a hash table.
To handle the non-manifold structure of $\msa{\hdim+1}$, we also determine the $\hdim$-simplices in $\nmj{\hdim}$ (those with more than two cofacets, see \cref{eq:nmk} and \cref{fig:3dmethod-concepts}).
Now, we can construct the \emph{Urquhart $(\hdim+1)$-cells}, delimited by $\usc{\hdim}\cup\nmj{\hdim}$ within $\msa{\hdim+1}$.
For that, like in the codimension-1 case, we maintain a union-find data structure over the $\msa{\hdim+1}$ simplices.
When encountering a simplex that is not in $\usc{\hdim}\cup\nmj{\hdim}$, 
we merge the two union-find classes associated with its (at most two) cofacets.
We also retrieve for each class the largest $(\hdim+1)$-simplex it contains.
Finally, a boundary matrix reduction allows us to obtain \PH{\hdim} by pairing the Urquhart $(\hdim+1)$-cells (the columns) with simplices in $\usc{\hdim}\cup\nmj{\hdim}$ (the rows). This corresponds to \cref{algo:intermediary-persistence}.

\paragraph{0- and 1-dimensional persistent pairs}

Additionally, in the specific case where $\hdim=1$, we apply Kruskal's algorithm on $\usc{1}=\ug$ to obtain $\mst$ (from which $\dgmdr{0}(\pc)$ can be deduced). Then, we can only consider the edges in $(\ug\cup\nmj{1})\setminus\mst$ as the rows of the boundary matrix that is reduced in order to pair \emph{Urquhart polygons} (the columns) with these edges. This corresponds to \cref{algo:1-persistence}.

\subsection{Parallelization}
\label{subsec:parallelization}
\cref{algo:intermediary-persistence,algo:1-persistence} can be parallelized. Apart from the easily parallelizable tasks (traversals, sorting), three steps in particular require special care.
First, establishing connectivity between $\hdim$- and $(\hdim-1)$-simplices in $\msa{\hdim}$ can be done in parallel using a concurrent hash table implementation.
Second, computing Urquhart cells can be done in parallel using a concurrent version of the union-find data structure, using wait-free algorithms that rely on CAS (compare-and-swap) operations~\cite{anderson1991wait}.
Finally, the boundary matrix reduction step that pairs $\hdim$-cells with $(\hdim-1)$-simplices can also be done in parallel~\cite{morozov2020towards, guillou2023discrete}: when a persistence pair $(\sigma,\tau)$ is created, it is temporary. If another thread later wants to pair $\sigma$ with a $\tau'$ such that $\tau'\prec\tau$, then $(\sigma,\tau)$ is updated into $(\sigma,\tau')$ and the reduction of $\tau$ restarts.

\section{Results}
\label{sec:results}

All results were obtained on a laptop computer with 64 GB of RAM and a Core i7-13850HX CPU (8 cores at 5.3 GHz and 12 cores at 3.8 GHz).
We implemented our method in C++.
For computing Delaunay complexes, we used the implementations of CGAL~\cite{fabri_cgal_2009}, either specialized 
for $\bbr^3$~\cite{cgal:pt-t3-24b} or for arbitrary
dimension~\cite{cgal:hdj-t-24b}. Only the specialization to $\bbr^3$ benefits from an
optional parallel implementation. However, this parallelization is not always advantageous, and can even be slower than the single-threaded version depending on the distribution of points in $\pc$.

\subsection{Runtime}
\label{sec:runtimes}

\paragraph{Compared methods}
We compared our approach with \emph{PHAT}~\cite{bauer2017phat} and \emph{Gudhi} \cite{maria2014gudhi}.
For \emph{PHAT}, the Delaunay triangulation is built with CGAL, then the boundary matrix is built with code adapted from an add-on in the \emph{PHAT} repository for computing persistence diagrams of $\alpha$-filtrations~\cite{phat_addon}, and finally reduced with the default algorithm and data structure provided by \emph{PHAT}.
For \emph{Gudhi}, a simplex tree containing the Delaunay complex is built with the $\alpha$-complex module~\cite{gudhi:AlphaComplex} (without computing minimum enclosing balls), then the filtration values are set on the edges (i.e., their length) before being propagated to the rest of the simplices, and the persistence diagram is computed with the persistent cohomology module~\cite{gudhi:PersistentCohomology}.
Finally, \emph{giotto-tda}~\cite{JMLR:v22:20-325} computes the Delaunay complex with the implementation of SciPy~\cite{virtanen2020scipy} and passes the sparse distance matrix containing Delaunay edges only to \emph{Ripser}~\cite{bauer_ripser_2021}. However, we did not include it in our comparison because it produces incorrect persistence diagrams and crashes when $n$ is too high due to an internal limitation on simplex indices.

\paragraph{Performance}
\cref{fig:timings-3} compares the runtime of these methods (ours, \emph{PHAT}, \emph{Gudhi}) for computing the persistence of the Delaunay--Rips filtration of point clouds in $\bbr^3$ distributed uniformly in the unit cube (left) and on the unit circle with some Gaussian noise (right).
To compare with the Delaunay filtration, we included the runtime of \emph{Gudhi} for this filtration.
To compare with the Rips filtration, we also measured the runtime of \emph{Ripser}~\cite{bauer_ripser_2021} for computing the 0-, 1-, and 2-dimensional persistence diagrams, showing---as expected---a much worse asymptotic behavior.
For completeness, we also measured the runtime of Euclidean PH1~\cite{koyama2023faster}, a method that computes 1-dimensional persistence diagrams of point clouds in $\bbr^3$.
\cref{fig:timings-d} does the same comparison on uniformly distributed point clouds in $\bbr^d$, $4\leq d\leq8$.
\cref{table:speedups} shows the average runtime ratio (with and without including the Delaunay complex computation step) and memory consumption ratio between our method and \emph{PHAT}, by dimension.
Overall, we observe that the runtime ratio between our method and \emph{PHAT} increases from 0.29 to 0.76 as $d$ increases from $3$ to $8$.
If we consider the Delaunay complex as the input and omit its computation in the runtime, the corresponding acceleration ratio varies from 0.20 to 0.73.
Thus, our acceleration becomes less important as the dimension increases.
Note that our implementation becomes noncompetitive for $d\geq9$.
In fact, $d=9$ is already substantial for computing Delaunay complexes (which, for only 200 points, requires approximately 40 seconds and 850 MB of memory).
Finally, our method uses between 0.19 and 0.35 times the memory used by \emph{PHAT}.

\begin{figure}
	\includegraphics[width=.49\linewidth]{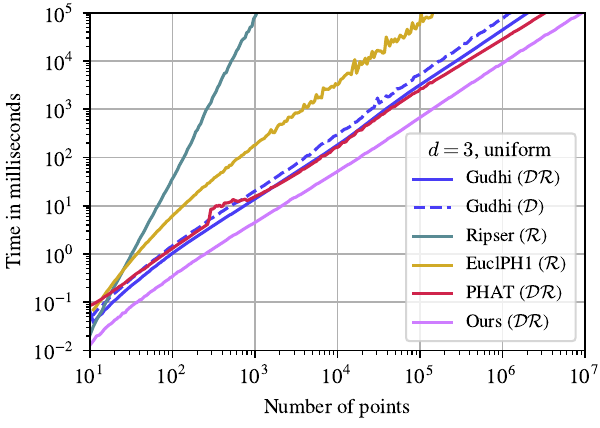}
	\includegraphics[width=.49\linewidth]{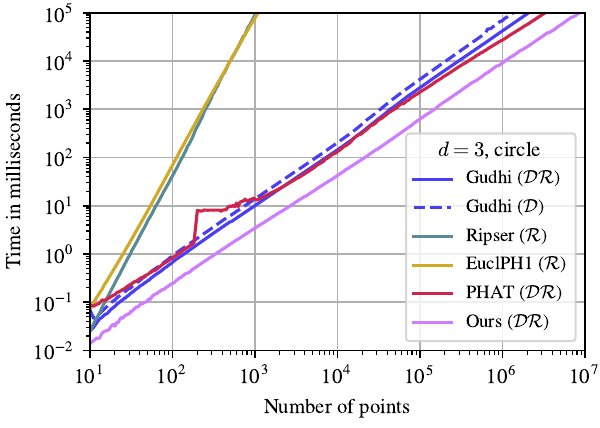}
	\caption{Runtime for computing the Delaunay--Rips persistence diagram of a point cloud in $\bbr^3$, uniformly distributed in the unit cube (left) and near the unit circle (right).
	We compare our approach described in \cref{sec:method} with \emph{Gudhi}~\cite{maria2014gudhi} and \emph{PHAT}~\cite{bauer2017phat}. The runtime for computing the Rips 0, 1 and 2-dimensional persistence diagrams with \emph{Ripser}~\cite{bauer_ripser_2021}, and that for computing the Rips 0 and 1-dimensional persistence diagrams with \emph{Euclidean PH1}, the method described in~\cite{koyama2023faster}, are also shown for comparison.}
	\label{fig:timings-3}
\end{figure}

\begin{figure}
	\includegraphics[height=.4\linewidth]{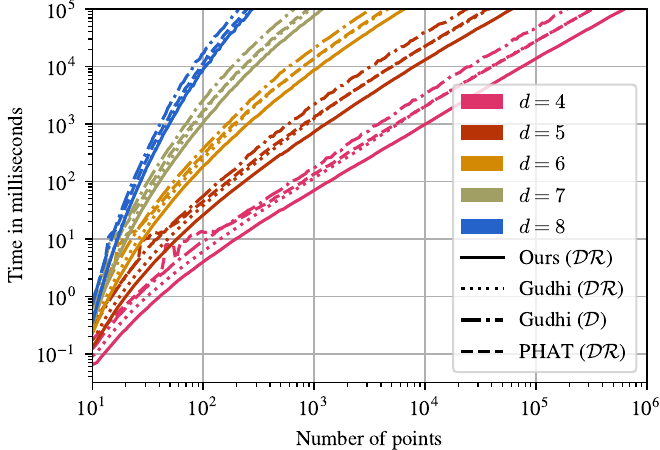}
	\hfill
	\includegraphics[height=.4\linewidth]{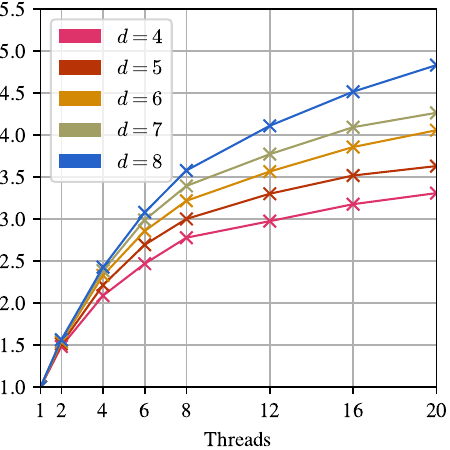}
	\caption{
	\textbf{Left:} Runtime for computing the Delaunay--Rips persistence diagram of a uniformly distributed point cloud in $\bbr^\ddim$, for varying $4\leq\ddim\leq8$. We compare our approach with \emph{Gudhi}~\cite{maria2014gudhi} and \emph{PHAT}~\cite{bauer2017phat}.
	\textbf{Right:} Speedup of our method (excluding the Delaunay computation step) according to the number of threads. Computed for uniformly distributed point sets (200,000 points for $d=4$, 20,000 points for $d=5$, 2000 points for $d=6$, 600 points for $d=7$, 200 points for $d=8$).
	Note that the CPU used for this experiment had a heterogeneous architecture with 8 fast cores (5.3 GHz) and 12 slower cores (3.8 GHz).}
	\label{fig:timings-d}
	\label{fig:timings-parallel}
\end{figure}

\begin{table}
	\centering
	\begin{tabular}{|l|c|c|c|c|c|c|c|}
		\hline
		Dimension $d$ & 3 & 4 & 5 & 6 & 7 & 8 & 9\\
		\hline
		Average overall timing ratio & 0.29 & 0.50 & 0.58 & 0.63 & 0.64 & 0.76 & 1.44 \\
		Average \PH{} timing ratio & 0.20 & 0.33 & 0.39 & 0.44 & 0.55 & 0.73 & 1.46 \\
		Average memory ratio & 0.35 & 0.26 & 0.19 & 0.19 & 0.20 & 0.19 & 0.24 \\
		\hline
	\end{tabular}
	\caption{Average runtime and memory consumption ratios of our method compared to \emph{PHAT}, by dimension. The \PH{} timing is obtained by removing the time required to compute the Delaunay complex.}
	\label{table:speedups}
\end{table}

\paragraph{Parallelism}
\cref{fig:timings-parallel} gives the speedup of the computation of the persistence diagrams with our parallel implementation (see \cref{subsec:parallelization}) according to the number of threads. Note that there is a single-thread bottleneck in the recovery of $\ddim$-simplices in the Delaunay complex (to which CGAL's implementation does not provide random access). In our settings, this speedup reached around 3 for 8 threads and around 4 for 20 threads. In addition, it tends to increase with the ambient dimension.


\paragraph{Number of simplices}

\cref{table:simplices_number} in \cref{appendix:number_of_simplices} gives the number of non-zero persistence pairs, the number of $\hdim$-cells, and the size of $\msa{\hdim}$, $\usc{\hdim}$ and $\del^\hdim$, with various ambient dimensions $\ddim$ and point cloud distributions (uniformly in $[-1,1]^\ddim$ or near the sphere $\bbs^{\ddim-1}$).
Examining the quantities $1-|\poly{\hdim}|/|\msa{\hdim}|$ (where $|\poly{\hdim}|$ is the number of $\hdim$-cells delimited by $\usc{\hdim-1}\cup\nmj{\hdim-1}$ within $\msa{\hdim}$) gives an idea of the proportion of column reductions that are avoided by considering the $\hdim$-cells instead of individual $\msa{\hdim}$ simplices, which allows for a faster reduction step.
In particular, one can observe that $|\poly{\hdim}|/|\msa{\hdim}|$ tends to increase with the ambient dimension $\ddim$, explaining the decreasing speedup noticed in \cref{table:speedups}, until it becomes unprofitable to compute these cells compared to the benefits it provides.

\subsection{Performance analysis in classification tasks}
\label{sec:classification}

To evaluate the practical value of the Delaunay--Rips filtration, we compare its performance against the Delaunay--\cech\ filtration (using the implementation of \emph{Gudhi}) in several classification tasks.

\paragraph{Methodology}
We vectorize persistence diagrams with several methods: ATOL~\cite{royer2021atol}, persistence landscapes (PL)~\cite{bubenik2015statistical} and persistence statistics vector (we took the formulation from the survey~\cite{ali2023survey}). These vectorizations are passed to the classifier LightGBM~\cite{ke2017lightgbm} with default parameters, and the machine learning pipelines are set up using the scikit-learn library~\cite{pedregosa2011scikit}.
Hyperparameters are validated with 5-fold cross-validation.
For ATOL, the number of $k$-means centroids is selected in $\{32,64,128\}$ and the contrast function (\cite{royer2021atol}, Eq. 1) is selected in $\{\text{Gaussian},\text{Laplacian},\text{indicator}\}$.
For persistence landscapes, the resolution is selected in \{$10\times10, 20\times20, 50\times50, 100\times100$\}. The persistence statistics vector does not have hyperparameters.

\paragraph{Datasets}

We begin by evaluating the filtrations on a synthetic toy dataset consisting of points uniformly sampled within the unit ball in $\bbr^3$, apart from a variable number $k\in\llbracket1,10\rrbracket$ of empty spherical regions.
Within a given point set, these voids share a common diameter $\delta$, chosen from $\{\delta_0,\delta_0+x/3,\delta_0+2x/3,\delta_0+x\}$, where $\delta_0=0.1$ and $x>0$.
The classification task is to identify which of the four diameter classes a given point set belongs to.
The parameter $x$ controls the separation between class diameters and therefore the overall task difficulty: larger values of $x$ make the classes easier to distinguish.
Classification relies solely on $H_2$.
Results and runtimes are shown in \cref{fig:synthetic}.

The second dataset models porous materials and is produced using the \texttt{PoreSpy} library~\cite{gostick2019porespy}.
It generates binary $64^3$ grids, which we then convert into point sets, adding Gaussian noise to avoid non-general positions arising from the lattice structure.
Porosity values are drawn uniformly at random from $[0.05,0.95]$.
The classification task is to determine which of two porous material generators produced a given sample:
\texttt{blobs}, which combines random noise with Gaussian blurring using a blobiness parameter drawn uniformly from $[0.5,2]$;
or \texttt{fractals}, which relies on simplex noise with a frequency drawn uniformly from $[0.05,0.2]$.
Here too, classification uses only $H_2$.
Results and runtimes are shown in \cref{fig:synthetic}.

Finally, we perform classification experiments on time series. TDA has previously been applied to such problems~\cite{perea2015sliding} via time-delay embeddings~\cite{takens2006detecting}. It is well established that signal properties (e.g., periodicity) manifest as specific topological features (e.g., loops) in the embedding.
We draw several datasets from the UCR archive~\cite{dau2019ucr}\footnote{available at \url{https://www.timeseriesclassification.com}}, spanning a range of sizes, lengths, class counts, and scientific domains (see \cref{table:ucr_datasets} for details), and use their recommended train and test splits.
Each series is embedded in $\bbr^3$ via the Takens time-delay embedding. This time, classification uses $H_0$ and $H_1$.
Results are presented in \cref{table:ucr_scores}, with runtimes shown in \cref{fig:ucr_timings}.

\paragraph{Results}

The results across the three experimental settings show that using the Delaunay--Rips filtration instead of the Delaunay--\cech\ filtration provides a substantial acceleration of not only the computation of persistence diagrams (see, e.g., \cref{fig:ucr_timings}, left), but also of the vectorization and classification processes (see, e.g., \cref{fig:ucr_timings}, right), since the diagrams tend to be smaller. The acceleration factor is roughly 5, which is compatible with \cref{sec:runtimes}.


\begin{figure}
	\centering
	\begin{minipage}{0.49\linewidth}
		\begin{tabular}{|ll|rrr|rr|}
			\hline
			&  & \multicolumn{3}{c|}{porous materials} & \multicolumn{2}{c|}{synthetic} \\
			\cline{3-7}
			&  & \multirow{2}{*}{accuracy} & \multicolumn{2}{c|}{runtime (s)} & \multicolumn{2}{c|}{runtime (s)} \\
			&  & & diag. & class. & diag. & class. \\
			\hline
			\multirow{3}{*}{$\delripsf{}$} & ATOL  & $0.885\pm0.027$ &  \multirow{3}{*}{24.0} & 17.4 & \multirow{3}{*}{5.7}  &  5.8 \\
			                               & PL    & $0.773\pm0.045$ &                          &  1.5 &                         &  1.4 \\
			                               & stats & $0.692\pm0.029$ &                        &  0.5 &                         &  0.7 \\
			\hline
			\multirow{3}{*}{$\delcechf{}$} & ATOL  & $0.921\pm0.028$ & \multirow{3}{*}{133.9} & 87.0 & \multirow{3}{*}{37.0} & 25.4 \\
			                               & PL    & $0.814\pm0.042$ &                          & 11.5 &                         &  4.4 \\
			                               & stats & $0.721\pm0.048$ &                        &  1.8 &                         &  1.2 \\
			\hline
		\end{tabular}
		\vspace{5mm}
	\end{minipage}
	\hspace{1.5cm}
	\begin{minipage}{0.4\linewidth}
		\includegraphics[width=\linewidth]{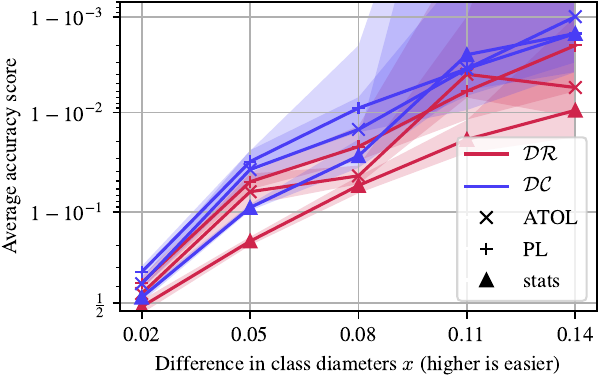}
	\end{minipage}
	\caption{Classification accuracy (averaged on 10 splits) using the Delaunay--Rips and Delaunay--\cech\ persistence diagrams with several vectorization methods, for the porous materials (left table) and the synthetic toy (right chart, parameterized by the $x$ value) datasets. The runtimes (in seconds) for computing the persistence diagrams and classifying from them are also shown in the table.}
	\label{fig:synthetic}
\end{figure}

However, on both synthetic datasets (\cref{fig:synthetic}), for a fixed vectorization method, the classification is slightly less accurate with Delaunay--Rips than with Delaunay--\cech. This is likely due to the classification problem requiring a good estimation of the size of the topological features, which is less precise with the instabilities of Delaunay--Rips persistence diagrams. In terms of vectorization, we notice that ATOL---and, to a lesser extent, PL---is more accurate than the persistence statistics vector, which loses information on individual, informative persistent pairs. 
Finally, for the real time-series datasets, the accuracies are generally similar with both filtrations (see \cref{table:ucr_scores}), and the persistence statistics vector often provides the best accuracy in both cases.

\begin{figure}
	\centering
	\includegraphics[width=.49\linewidth]{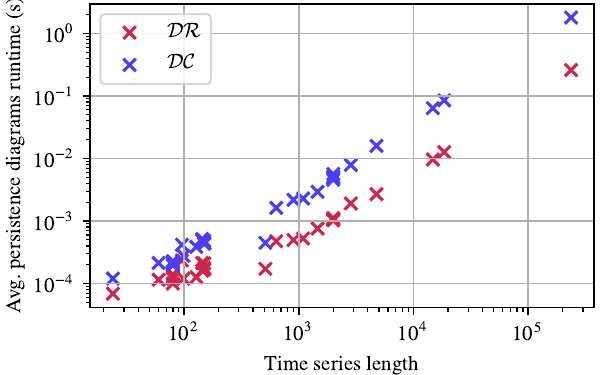}
	\includegraphics[width=.49\linewidth]{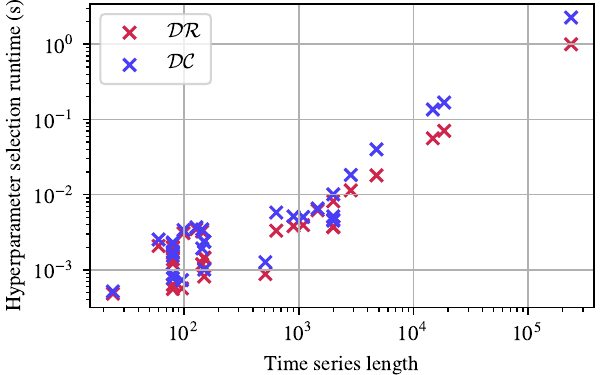}
	\caption{Runtimes of persistence diagram computation (left) and average classification time from the persistence diagrams (right) for Delaunay--Rips and Delaunay--\cech\ filtrations with ATOL as vectorization.}
	\label{fig:ucr_timings}
\end{figure}

\section*{Conclusion}
\label{sec:conclusion}

In this work, we performed a theoretical and empirical study of the behavior of persistence diagrams for the Delaunay--Rips filtration of point clouds. In particular, we studied the stability of the associated persistence diagrams in terms of bottleneck distance, and how they approximate the more common Rips persistence diagrams.
Then, we proposed a fast and memory-efficient algorithm that computes these Delaunay--Rips persistence diagrams of low-dimensional point clouds, leveraging the specific structure of the filtration.
This algorithm, supported by our theoretical and experimental study, makes the Delaunay--Rips filtration a relevant candidate in practice for the fast computation of persistent homology of point clouds in low-dimensional settings.
In particular, it can serve as a way to accelerate persistence-based topological machine learning approaches involving low-dimensional input point clouds or low-dimensional latent spaces.
As an illustration, we have evaluated its ability to classify synthetic and real datasets faster than other filtrations with a limited loss of accuracy.
In future work, we aim to generalize our results to other flag filtrations, with the goal of achieving faster computations by either efficiently discarding irrelevant simplices or merging them into fewer, larger cells.

\subsubsection*{Acknowledgments}

This work is partially supported by the European Commission grant ERC-2019-COG \emph{``TORI''} (ref. 863464, \url{https://erc-tori.github.io/}).
We gratefully acknowledge support from the CNRS/IN2P3 Computing Center (Lyon - France) for providing computing and data-processing resources needed for this work.

\bibliography{biblio}
\bibliographystyle{tmlr}

\appendix

\section{Proofs}
\label{appendix:proofs}

\begin{proof}[\textbf{Proof of \cref{lemma:rips_delrips}}]
	Using the induced inclusions (\cref{eq:inclusions_jung,eq:inclusions_del}) and induced isomorphisms (\cref{eq:induced_isomorphism}) between homology groups, we can define linear maps $\maprtodr{r}{k}:\hgf{\hdim}{\ripsn}{r}\rightarrow\hgf{\hdim}{\delripsn}{\jung{\hdim+1}r}$ as:
	\begin{center}
		\begin{tikzcd}[cramped, sep=scriptsize] 
			\maprtodr{r}{k}:\hgf{\hdim}{\ripsn}{r} \arrow[r]
			&\hgf{\hdim}{\cechn}{\jung{\hdim+1}r} \arrow[from=r,"\iota^\star"',shift right] \arrow[r,"(\iota^\star)^{-1}"',shift right]
			&\hgf{\hdim}{\delcechn}{\jung{\hdim+1}r} \arrow[r]
			&\hgf{\hdim}{\delripsn}{\jung{\hdim+1}r}
		\end{tikzcd}.
	\end{center}
	Now, to show that the persistence modules
	$\big(\hgf{\hdim}{\ripsn}{r}\big)_{r\geq0}$ and
	$\big(\hgf{\hdim}{\delripsn}{r}\big)_{r\geq0}$ are multiplicatively
	$\jung{\hdim+1}$-interleaved, it suffices to notice that the following
	diagrams (on the homology groups level) commute:
	\begin{center}
		\begin{tabular}{cc}
			\begin{tikzcd}[cramped, row sep=scriptsize, column sep=small]
				\hgf{\hdim}{\delripsn}{r} \arrow[r] \arrow[d] & \hgf{\hdim}{\delripsn}{\jung{\hdim+1}r} \\
				\hgf{\hdim}{\ripsn}{r} \arrow[ur,"\maprtodr{r}{k}"']
			\end{tikzcd}
			& 
			\begin{tikzcd}[cramped, row sep=scriptsize, column sep=small]
				& \hgf{\hdim}{\delripsn}{\jung{\hdim+1}r} \arrow[r] & \hgf{\hdim}{\delripsn}{\jung{\hdim+1}r'} \\
				\hgf{\hdim}{\ripsn}{r} \arrow[r] \arrow[ur,"\maprtodr{r}{k}"] & \hgf{\hdim}{\ripsn}{r'} \arrow[ur,"\maprtodr{r'}{k}"'] 
			\end{tikzcd}
			\\\\\\
			\begin{tikzcd}[cramped, row sep=scriptsize, column sep=small]
				& \hgf{\hdim}{\delripsn}{r} \arrow[d] \\
				\hgf{\hdim}{\ripsn}{r/\jung{\hdim+1}} \arrow[r] \arrow[ur,"\maprtodr{r/\jung{\hdim+1}}{k}"] & \hgf{\hdim}{\ripsn}{r}
			\end{tikzcd}
			&
			\begin{tikzcd}[cramped, row sep=scriptsize, column sep=small]
				\hgf{\hdim}{\delripsn}{r} \arrow[r] \arrow[d] & \hgf{\hdim}{\delripsn}{r'} \arrow[d] \\
				\hgf{\hdim}{\ripsn}{r} \arrow[r] & \hgf{\hdim}{\ripsn}{r'}
			\end{tikzcd}
		\end{tabular}
	\end{center}
	
	To see that, one can decompose the diagrams following the expression of $\maprtodr{r}{k}$. For instance, for the first diagram, 
	\begin{center}
		\begin{tikzcd}[cramped, row sep=scriptsize, column sep=scriptsize]
			\hgf{\hdim}{\delripsn}{r} \arrow[r] \arrow[ddd] \arrow[ddr,dashed] \arrow[dr,dashed] & \hgf{\hdim}{\delripsn}{\jung{\hdim+1}r} \\
			& \hgf{\hdim}{\delcechn}{\jung{\hdim+1}r} \arrow[u] \arrow[d,"\simeq"] \\
			& \hgf{\hdim}{\cechn}{\jung{\hdim+1}r} \\
			\hgf{\hdim}{\ripsn}{r} \arrow[ur]
		\end{tikzcd}
	\end{center}
	commutes because its counterpart on the simplicial level commutes, and the induced isomorphic inclusion can be inverted while maintaining commutativity. The same applies to the second and third diagrams.

	Finally, the bound given by \cref{eq:bottle_logdgms} is tight, since it is possible to find $X\subset\bbr^{\ddim+1}$ such that $\bottle\left(\log\dgmdr{\ddim}(X),\log\dgmrips{\ddim}(X)\right)$ is arbitrarily close to its upper bound $\log\jung{\ddim+1}$.
	For that, it suffices to take two antipodal points, e.g., $x_0 =(1,0,\ldots,0)$ and $x_1=(-1,0,\ldots,0)$ that will be joined by a Delaunay edge, the points of a regular $(\ddim+1)$-simplex inscribed in $(1+\varepsilon)\bbs^\ddim$, and finally a sufficiently dense sampling of $(1+\varepsilon)\bbs^\ddim$ to ensure the existence of a $\PH{\ddim}$ class.
	It will be killed by the Delaunay edge $\{x_0,x_1\}$ of length 2 in the Delaunay--Rips filtration, and by an edge of the regular simplex of length $(1+\varepsilon)\sqrt{2(\ddim+2)/(\ddim+1)}=2(1+\varepsilon)/\jung{\ddim+1}$ in the Rips filtration.
Such configurations are shown for $\bbr^2$ and $\bbr^3$ in \cref{fig:worstCases}, where the points are arranged respectively as a regular hexagon and a regular dodecahedron.
\end{proof}

\begin{figure}
	\centering
	\begin{tabular}{cc}
		$\pc\subset\bbr^2$ & $\pc\subset\bbr^3$\\
		\def\svgwidth{.3\linewidth}
\begingroup%
  \makeatletter%
  \providecommand\color[2][]{%
    \errmessage{(Inkscape) Color is used for the text in Inkscape, but the package 'color.sty' is not loaded}%
    \renewcommand\color[2][]{}%
  }%
  \providecommand\transparent[1]{%
    \errmessage{(Inkscape) Transparency is used (non-zero) for the text in Inkscape, but the package 'transparent.sty' is not loaded}%
    \renewcommand\transparent[1]{}%
  }%
  \providecommand\rotatebox[2]{#2}%
  \newcommand*\fsize{\dimexpr\f@size pt\relax}%
  \newcommand*\lineheight[1]{\fontsize{\fsize}{#1\fsize}\selectfont}%
  \ifx\svgwidth\undefined%
    \setlength{\unitlength}{296.55363068bp}%
    \ifx\svgscale\undefined%
      \relax%
    \else%
      \setlength{\unitlength}{\unitlength * \real{\svgscale}}%
    \fi%
  \else%
    \setlength{\unitlength}{\svgwidth}%
  \fi%
  \global\let\svgwidth\undefined%
  \global\let\svgscale\undefined%
  \makeatother%
  \begin{picture}(1,0.95582278)%
    \lineheight{1}%
    \setlength\tabcolsep{0pt}%
    \put(0,0){\includegraphics[width=\unitlength,page=1]{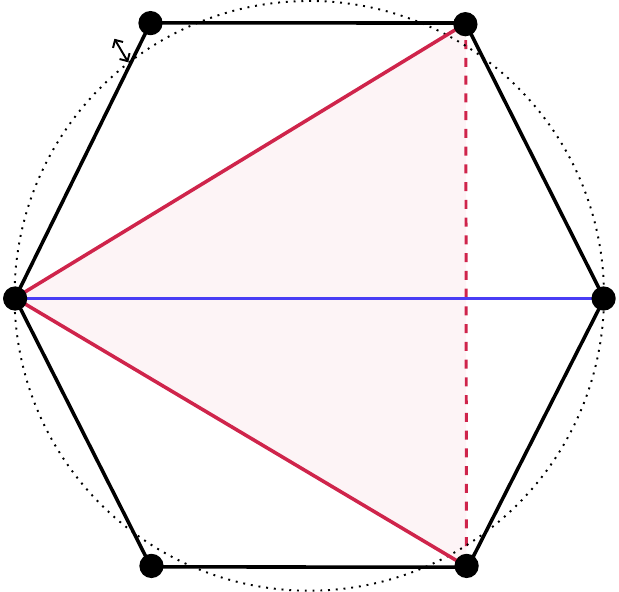}}%
    \put(0.53557776,0.49550591){\color[rgb]{0.28627451,0.23921569,0.96078431}\makebox(0,0)[t]{\lineheight{1.25}\smash{\begin{tabular}[t]{c}$2$\end{tabular}}}}%
    \put(0.16125385,0.84813123){\color[rgb]{0,0,0}\makebox(0,0)[t]{\lineheight{1.25}\smash{\begin{tabular}[t]{c}$\varepsilon$\end{tabular}}}}%
  \end{picture}%
\endgroup%
&
		\includegraphics[width=.3\linewidth]{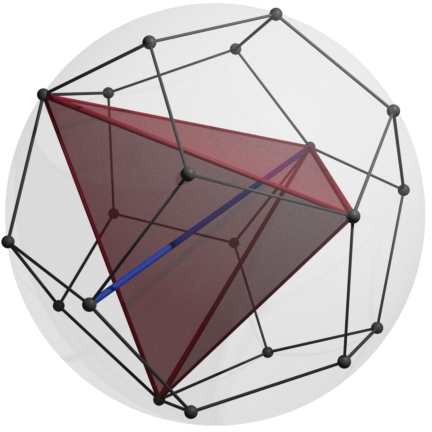}\\
		$\color{customBlue}\dgmdr{1}(\pc)\simeq\left\{\left(\frac{1}{2},1\right)\right\}$ &
		$\color{customBlue}\dgmdr{2}(\pc)\simeq\left\{\left(\frac{1}{\sqrt{3}},1\right)\right\}$ \\
		$\color{customRed}\dgmrips{1}(\pc)\simeq\left\{\left(\frac{1}{2},\frac{\sqrt{3}}{2}\right)\right\}$ & $\color{customRed}\dgmrips{2}(\pc)\simeq\left\{\left(\frac{1}{\sqrt{3}},\frac{2}{\sqrt{6}}\right)\right\}$
	\end{tabular}
	\vspace{.2cm}
	\caption{
		\textbf{Left}: $\pc\subset\bbr^2$ such that $\bottle\left(\log\dgmdr{1}(\pc),\log\dgmrips{1}(\pc)\right)$ is arbitrarily close to $\log\frac{2}{\sqrt{3}}=\log\jung{2}$. The points are arranged as a regular hexagon inscribed in the unit circle, slightly perturbed so that its Delaunay triangulation (solid edges) contains a diametrical edge (in blue). The red dotted edge that exists in $\Delta(\pc)$ but not in $\del(\pc)$ makes the \PH{1} class less persistent in the Rips filtration than in the Delaunay--Rips filtration.
		\textbf{Right}: $\pc\subset\bbr^3$ such that $\bottle\left(\log\dgmdr{2}(\pc),\log\dgmrips{2}(\pc)\right)$ is arbitrarily close to $\log\frac{\sqrt{6}}{2}=\log\jung{3}$. The points are arranged as a regular dodecahedron inscribed in the unit sphere, slightly perturbed so that its Delaunay triangulation contains a diametrical edge (in blue). A non-Delaunay regular tetrahedron (in red) makes the \PH{2} class less persistent in the Rips filtration than in the Delaunay--Rips filtration.
	}
	\label{fig:worstCases}
\end{figure}

\begin{proof}[\textbf{Proof of \cref{lemma:delrips_instab}}]
	Let $\hdim\geq0$ be a fixed homology dimension and $\eta=\jung{\hdim+1}$.
	Let also $\varepsilon>2\dgh(X,Y)$.
	We know that $\Big(H\big(\cechf{r}(X)\big)\Big)_{r\geq0}$ and
	$\Big(H\big(\cechf{r}(Y)\big)\Big)_{r\geq0}$ are $\varepsilon$-interleaved (see
	Lemma 4.4 of \cite{chazal2014persistence}). We write
	$\mapgeoxtoy{r}{\hdim}:H(\cechf{r}(X))\rightarrow
	H\big(\cechf{r+\varepsilon}(Y)\big)$ and
	$\mapgeoytox{r}{\hdim}:H(\cechf{r}(Y))\rightarrow
	H\big(\cechf{r+\varepsilon}(X)\big)$ the corresponding linear maps.
	In addition to the induced inclusions (\cref{eq:inclusions_jung,eq:inclusions_del}) and induced isomorphisms (\cref{eq:induced_isomorphism}), they permit us to define the following linear maps (we write $\hgdr{\hdim}{\delripsn}{r}{X}$ instead of $\hgfx{\hdim}{\delripsn}{r}{X}$ to simplify notation): 
	\begin{center}
		\begin{tikzcd}[cramped,sep=scriptsize]
			\mapxtoy{r}{k}:\hgdr{\hdim}{\delripsn}{r}{X} \arrow[r]
			&\hgfx{\hdim}{\cechn}{\eta r}{X} \arrow[r,"\mapgeoxtoy{\eta r}{k}"]
			&\hgfx{\hdim}{\cechn}{\eta r+\varepsilon}{Y} \arrow[from=r,"\iota^\star"',shift right] \arrow[r,"(\iota^\star)^{-1}"',shift right]
			&\hgfx{\hdim}{\delcechn}{\eta r+\varepsilon}{Y} \arrow[r]
			&\hgdr{\hdim}{\delripsn}{\eta r+\varepsilon}{Y}\\
			\mapytox{r}{k}:\hgdr{\hdim}{\delripsn}{r}{Y} \arrow[r]
			&\hgfx{\hdim}{\cechn}{\eta r}{Y} \arrow[r,"\mapgeoytox{\eta r}{k}"]
			&\hgfx{\hdim}{\cechn}{\eta r+\varepsilon}{X} \arrow[from=r,"\iota^\star"',shift right] \arrow[r,"(\iota^\star)^{-1}"',shift right]
			&\hgfx{\hdim}{\delcechn}{\eta r+\varepsilon}{X} \arrow[r]
			&\hgdr{\hdim}{\delripsn}{\eta r+\varepsilon}{X}
		\end{tikzcd}
	\end{center}
	Now, we notice that the following diagrams involving these two maps commute:
	\begin{center}
		\begin{tabular}{cc}
			\begin{tikzcd}[cramped,row sep=scriptsize, column sep=small]
				\hgdr{\hdim}{\delripsn}{(r-\varepsilon)/\eta}{X} \arrow[rr] \arrow[dr,"\mapxtoy{(r-\varepsilon)/\eta}{k}"'] & & \hgdr{\hdim}{\delripsn}{\eta r+\varepsilon}{X} \\
				& \hgdr{\hdim}{\delripsn}{r}{Y} \arrow[ur,"\mapytox{r}{k}"']
			\end{tikzcd}
			& 
			\begin{tikzcd}[cramped,row sep=scriptsize, column sep=small]
				& \hgdr{\hdim}{\delripsn}{\eta r+\varepsilon}{X} \arrow[r] & \hgdr{\hdim}{\delripsn}{\eta r'+\varepsilon}{X} \\
				\hgdr{\hdim}{\delripsn}{r}{Y} \arrow[r] \arrow[ur,"\mapytox{r}{k}"] & \hgdr{\hdim}{\delripsn}{r'}{Y} \arrow[ur,"\mapytox{r'}{k}"'] 
			\end{tikzcd}
			\\\\\\
			\begin{tikzcd}[cramped,row sep=scriptsize, column sep=small]
				& \hgdr{\hdim}{\delripsn}{r}{X} \arrow[dr,"\mapxtoy{r}{k}"] \\
				\hgdr{\hdim}{\delripsn}{(r-\varepsilon)/\eta}{Y} \arrow[rr] \arrow[ur,"\mapytox{(r-\varepsilon)/\eta}{k}"] & & \hgdr{\hdim}{\delripsn}{\eta r+\varepsilon}{Y}
			\end{tikzcd}
			&
			\begin{tikzcd}[cramped,row sep=scriptsize, column sep=small]
				\hgdr{\hdim}{\delripsn}{r}{X} \arrow[dr,"\mapxtoy{r}{k}"'] \arrow[r] & \hgdr{\hdim}{\delripsn}{r'}{X} \arrow[dr,"\mapxtoy{r'}{k}"] \\
				& \hgdr{\hdim}{\delripsn}{\eta r+\varepsilon}{Y} \arrow[r] & \hgdr{\hdim}{\delripsn}{\eta r'+\varepsilon}{Y}
			\end{tikzcd}
		\end{tabular}
	\end{center}
	To see that, as in \cref{lemma:rips_delrips}, we can decompose the diagrams following the expressions of $\mapxtoy{}{}$ and $\mapytox{}{}$. For instance, for the first diagram,
	\begin{center}
		\begin{tikzcd}[cramped,row sep=scriptsize, column sep=scriptsize]
			|[xshift=-5.7em]| \hgfx{\hdim}{\delripsn}{(r-\varepsilon)/\eta}{X} \arrow[r] \arrow[dd] & |[xshift=5.7em]| \hgfx{\hdim}{\delripsn}{\eta r+\varepsilon}{X} \\
			& |[xshift=5.7em]| \hgfx{\hdim}{\delcechn}{\eta r+\varepsilon}{X} \arrow[u] \arrow[d,"\simeq"]\\
			|[xshift=-2em]| \hgfx{\hdim}{\cechn}{r-\varepsilon}{X} \arrow[d,"\mapgeoxtoy{r-\varepsilon}{\hdim}"'] \arrow[r,dashed] & |[xshift=5.7em]| \hgfx{\hdim}{\cechn}{\eta r+\varepsilon}{X} \\
			\hgfx{\hdim}{\cechn}{r}{Y} \arrow[r,dashed] & |[xshift=3.7em]| \hgfx{\hdim}{\cechn}{\eta r}{Y} \arrow[u,"\mapgeoytox{\eta r}{\hdim}"'] \\
			\hgfx{\hdim}{\delcechn}{r}{Y} \arrow[u,"\simeq"] \arrow[d] \\
			\hgfx{\hdim}{\delripsn}{r}{Y} \arrow[uur]
		\end{tikzcd}
	\end{center}
	commutes: for the top and bottom parts, because it is the case for its
counterpart on the simplicial level; for the central part, because this
corresponds to the $\varepsilon$-interleaving between
$\Big(H\big(\cechf{r}(X)\big)\Big)_{r\geq0}$ and
$\Big(H\big(\cechf{r}(Y)\big)\Big)_{r\geq0}$. Then the isomorphic inclusions
can be inverted while maintaining commutativity.
	
	Now, we rescale those persistence modules so that we get a usable additive interleaving.
	Let $f:r\mapsto\log\left(r+\frac{\varepsilon}{\eta-1}\right)$, which is chosen so that:
	\begin{align*}
		f(\eta r+\varepsilon)&=\log\left(\eta r+\varepsilon+\frac{\varepsilon}{\eta-1}\right)\\
		&=\log\left(\eta\left(r+\frac{\varepsilon}{\eta-1}\right)\right)\\
		&=f(r)+\log\eta
	\end{align*}
	Therefore, with the four commutative diagrams above, the modules
	$\big(\hgdr{\hdim}{\delripsn}{f^{-1}(\alpha)}{X}\big)_{\alpha}$ and
$\big(\hgdr{\hdim}{\delripsn}{f^{-1}(\alpha)}{Y}\big)_{\alpha}$ are
additively $(\log\eta)$-interleaved. With the strong stability
theorem~\cite{chazal2009proximity}, we obtain that the images through $f$ of the
diagrams $\dgmdr{\hdim}(X)$ and $\dgmdr{\hdim}(Y)$ have their bottleneck distance bounded
by $\log\eta$, i.e.,
	\begin{equation}
\bottle\Big(f\big(\dgmdr{\hdim}(X)\big),f\big(\dgmdr{\hdim}(Y)\big)\Big)\leq\log\eta
\text { where } f(D)=\{(f(b),f(d))\mid(b,d)\in D\}.
		\label{eq:bottle_f_diags}
	\end{equation}
	
	Finally, let $p=(x,y)\in\dgmdr{\hdim}(X)\cup\Delta$ and $p'=(x',y')\in\dgmdr{\hdim}(Y)\cup\Delta$, where $\Delta=\{(x,x)\mid x\geq0\}$ is the diagonal, such that $f(p)$ and $f(p')$ are paired together in an optimal pairing for \cref{eq:bottle_f_diags}. We have for $x$ and $x'$ (this is the same for $y$ and $y'$): 
	\begin{align*}
		|f(x)-f(x')|&\leq\log\eta\\
		\left|\log\left(x+\frac{\varepsilon}{\eta-1}\right)-\log\left(x'+\frac{\varepsilon}{\eta-1}\right)\right|&\leq\log\eta\\
		\left|\log\frac{x+\frac{\varepsilon}{\eta-1}}{x'+\frac{\varepsilon}{\eta-1}}\right|&\leq\log\eta\\
		\frac{1}{\eta}\leq\frac{x+\frac{\varepsilon}{\eta-1}}{x'+\frac{\varepsilon}{\eta-1}}&\leq\eta
	\end{align*}
	Therefore:
	\begin{align*}
		x+\frac{\varepsilon}{\eta-1}\leq\eta x'+\frac{\eta\varepsilon}{\eta-1}\quad&\text{and}\quad x'+\frac{\varepsilon}{\eta-1}\leq\eta x+\frac{\eta\varepsilon}{\eta-1}\\
		x-x'\leq(\eta-1)x'+\varepsilon\quad&\text{and}\quad x'-x\leq(\eta-1)x+\varepsilon
	\end{align*}
	Thus, $|x-x'|\leq(\eta-1)\max\{x,x'\}+\varepsilon$, and similarly for $y$ and $y'$.
	Hence,
	\[\lVert p-p'\rVert_\infty\leq(\eta-1)\max\{\lVert p\rVert_\infty,\lVert p'\rVert_\infty\}+\varepsilon\leq(\eta-1)\max\{\delta(X),\delta(Y)\}+\varepsilon\]
	since coordinates in a Delaunay--Rips persistence diagram correspond to edge lengths. Thus $\lVert p\rVert_\infty\leq\delta(X)$ for all $p\in\dgmdr{\hdim}(X)$.
	Therefore, using such an optimal pairing for \cref{eq:bottle_f_diags} in
order to pair $\dgmdr{\hdim}(X)$ and $\dgmdr{\hdim}(Y)$ together shows that
	$\bottle\big(\dgmdr{\hdim}(X),\dgmdr{\hdim}(Y)\big)
	\leq(\eta-1)\max\{\delta(X),\delta(Y)\} + \varepsilon$, which concludes the proof.
\end{proof}

\begin{proof}[\textbf{Proof of \cref{lemma:msa_in_usc}}]
    Let $\sigma\in\msa{\hdim}$. Suppose that $\sigma\notin\usc{\hdim}$.
    By definition, there exists $\tau\in\cof(\sigma, \del)$ such that $\sigma=\max\partial\tau$.
   	Now, there exists a $\sigma'\in\partial\tau\setminus\{\sigma\}$ such that $\sigma'\notin\msa{\hdim}$ (otherwise $\partial\tau$ would constitute a $\hdim$-cycle in $\msa{\hdim}$, which is acyclic).
   	According to the exchange property of spanning acycles (\cref{lemma:exchange}), $\msa{\hdim}\cup\{\sigma'\}\setminus\{\sigma\}$ is also a $\hdim$-spanning acycle, with $\sigma'\prec\sigma$, which contradicts the minimality of $\msa{\hdim}$.
\end{proof}

\begin{proof}[\textbf{Proof of \cref{lemma:non_apparent_in_usc}}]
	Let $\sigma\in\del^{\hdim}$ such that $\sigma\notin\usc{\hdim}$.
	By definition, there exists at least one $\tau\in\cof(\sigma,\del)$ such that $\sigma=\max\partial\tau$.
	Let $\tau^\star$ be the first of them that appears in the filtration
	(i.e., the smallest for $\prec$).
	Then, by definition, $(\sigma,\tau^\star)$ is an apparent $\PH{\hdim}$ pair, and therefore a zero persistence pair.
	Therefore, a $\hdim$-simplex that gives birth to a positive $\PH{\hdim}$ pair is necessarily in $\usc{\hdim}$.
	In addition, a $\hdim$-simplex in $\msa{\hdim}$ corresponds to the death of a $\PH{\hdim-1}$ pair, therefore not to the birth of a $\PH{\hdim}$ pair, hence the result.  
\end{proof}

\begin{proof}[\textbf{Proof of \cref{lemma:ph1_ug_is_positive}}]
	Recall that $\usc{1}=\ug$ and $\msa{1}=\mst$.
	Hence, \cref{lemma:non_apparent_in_usc} gives that the edges giving birth to a positive \PH{1} pair are in $\ug\setminus\mst$.  
	Conversely, if $(\sigma,\tau)$ is a zero persistence \PH{1} pair, \cref{lemma:ph1_apparent_is_zero} states that it is an apparent pair. Then $\sigma\subset\tau$ and $\sigma=\max\partial\tau$.
	Therefore $\sigma\notin\ug$, hence the result.
\end{proof}

\begin{proof}[\textbf{Proof of \cref{lemma:usc_in_msa}}]
	Suppose by contradiction that we have a $\hdim$-simplex $\sigma$ that satisfies this condition, but such that $\sigma\notin\usc{\hdim}$.
	This means that there exists a $\tau^\star\in\cof(\sigma,\del)$ such that $\sigma=\max\partial\tau^\star$.
	However, by hypothesis, any $\tau\in\cof(\sigma,\msa{\hdim+1})$ satisfies $\sigma\prec\max\partial\tau$, which implies here that $\tau^\star\not\in\msa{\hdim+1}$.

	In addition, as $\msa{\hdim+1}$ is spanning, there exists a $\tau'\in\cof(\sigma,\msa{\hdim+1})$. Indeed, take any of the cofacets $\tau\in\cof(\sigma,\del)$ in the Delaunay complex. By definition of the spanning-ness, the $\hdim$-cycle $\partial\tau$ has to be the boundary of some $(\hdim+1)$-chain in $\msa{\hdim+1}$, which gives such a $\tau'\in\cof(\sigma,\msa{\hdim+1})$.

	Finally, by hypothesis, this $\tau'$ satisfies $\sigma\prec\max\partial\tau'$.
	Thus, we have $\max\partial\tau^\star\prec\max\partial\tau'$, hence $\tau^\star\prec\tau'$ (see \cref{eq:def_order}).
	With the exchange property (\cref{lemma:exchange}), $\msa{\hdim+1}\cup\{\tau^\star\}\setminus\{\tau'\}$ is a $(\hdim+1)$-spanning acycle with $\tau^\star\prec\tau'$, which contradicts the minimality of $\msa{\hdim+1}$.
\end{proof}

\section{Empirical evaluation of approximation and stability}
\label{sec:distance-stability-practice}

We show in \cref{fig:hist_logbottle} the distribution of the bottleneck distance between Rips and Delaunay--Rips logarithmic 1- and 2-dimensional persistence diagrams for point sets in $\bbr^3$ with several distributions. This illustrates that the bound \cref{eq:bottle_logdgms} is pessimistic for uniformly sampled points (left), while \PH{1} and \PH{2} bounds can be closely reached for points sampled respectively on a circle (center) and a sphere (right), which correspond to the critical configurations described in \cref{sec:comp_r_dr} and \cref{fig:worstCases}.

In \cref{fig:hist_stability} we compare the empirical stability of Rips (\cref{eq:ripsstab}) and Delaunay--Rips (\cref{eq:bottle_stab2}) persistence diagrams to perturbations of the input. Unlike the Rips filtration, using the Delaunay--Rips filtration can lead to a bottleneck distance beyond $2\dgh(X,Y)$ due to the term $\max\{\delta(X),\delta(Y)\}$ in \cref{eq:bottle_stab2} that depends on $X$ and $Y$ directly rather than $\dgh(X,Y)$. Once again, this instability is worse for points distributed close to a circle or a sphere as this corresponds to the critical configurations.

In summary, the stability of Delaunay--Rips persistence diagrams should not be assumed in applications, even though discontinuities informally depend on the presence of very distinct topological features (e.g., a dense sampling exactly on a circle or a sphere).
However, given a point cloud, Delaunay--Rips persistence diagrams still capture topological information that remains close to that captured by Rips persistence diagrams, up to the size of the point cloud. This suggests an explanation for its usability for diverse applications such as in \cref{sec:classification} or those mentioned in~\cite{gabrielsson2020topology} (topological prior incorporation) or~\cite{mishra2023stability} (classification) which do not critically depend on a stability assumption.

\begin{figure}
	\includegraphics[width=.329\linewidth]{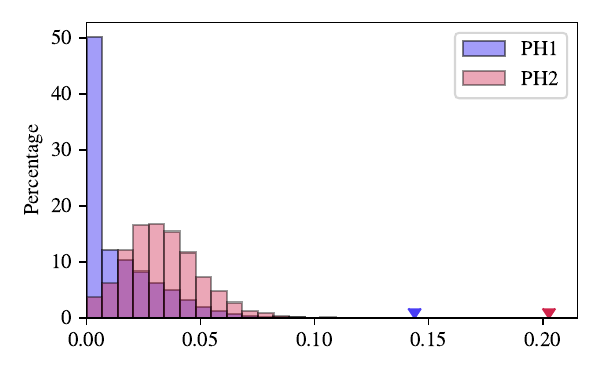}
	\includegraphics[width=.329\linewidth]{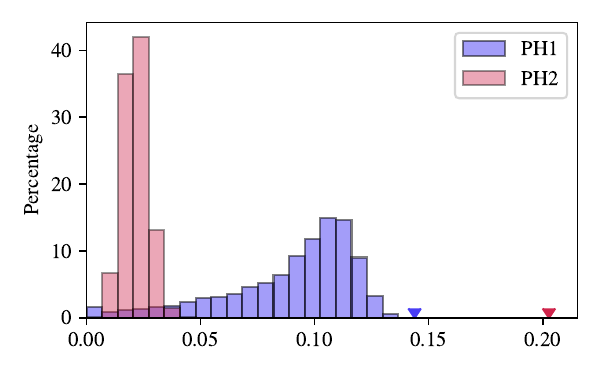}
	\includegraphics[width=.329\linewidth]{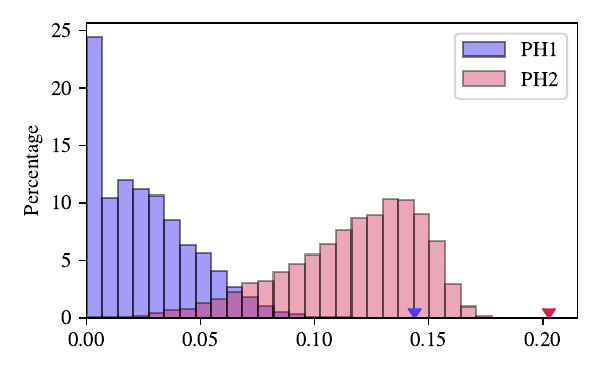}
	\caption{
		Histograms of the distribution of the bottleneck distance between
		logarithmic persistence diagrams for the Rips and Delaunay--Rips filtrations.
		Triangles indicate the theoretical upper bound of these distances for
		a given dimension of homology.
		\textbf{Left:} 10,000 sets of 50 points uniformly distributed in $[0,1]^3$.
		\textbf{Center:} 10,000 sets of 50 points distributed on a unit circle embedded in $\bbr^3$ with Gaussian noise.
		\textbf{Right:} 10,000 sets of 50 points uniformly distributed on a unit 2-sphere with Gaussian noise.
	}
	\label{fig:hist_logbottle}
\end{figure}

\begin{figure}
	\hspace{.2cm}
	\begin{tabular}{ccccc}
		& & $[-1,1]^3$ & $\bbs^1$ & $\bbs^2$ \\
		\raisebox{.2cm}{\multirow{2}{*}{\rotatebox{90}{Rips}}} &
		\raisebox{.3cm}{\scalebox{.55}{\rotatebox{90}{$\bottle(\dgmrips{1}(X),\dgmrips{1}(Y))$}}} &
		\hspace{-1mm}\includegraphics[width=.25\linewidth]{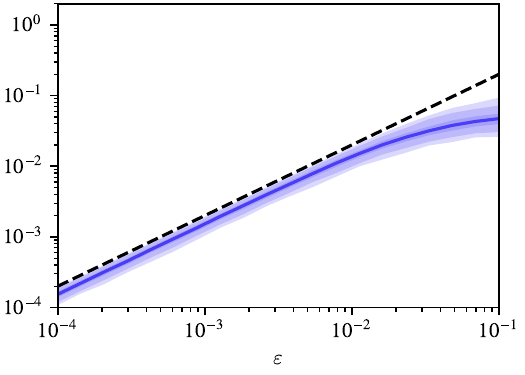}&
		\hspace{-2mm}\includegraphics[width=.25\linewidth]{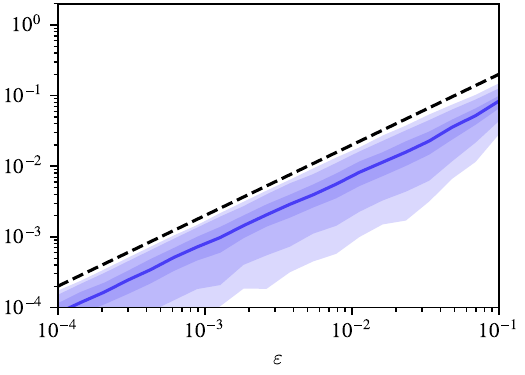}&
		\hspace{-2mm}\includegraphics[width=.25\linewidth]{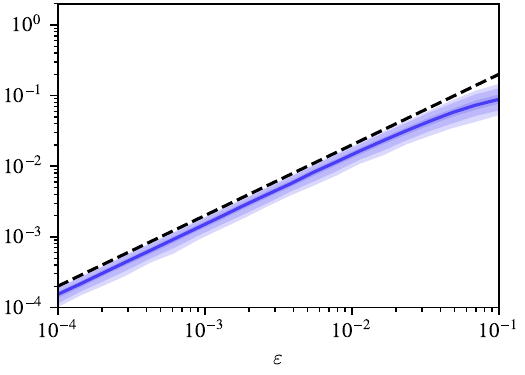}\\
		 &
		\raisebox{.3cm}{\scalebox{.55}{\rotatebox{90}{$\bottle(\dgmrips{2}(X),\dgmrips{2}(Y))$}}} &
		\hspace{-1mm}\includegraphics[width=.25\linewidth]{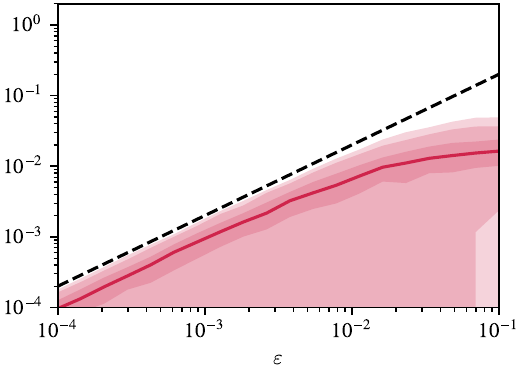}&
		\hspace{-2mm}\includegraphics[width=.25\linewidth]{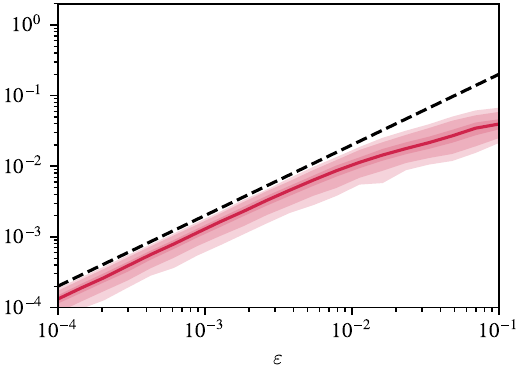}&
		\hspace{-2mm}\includegraphics[width=.25\linewidth]{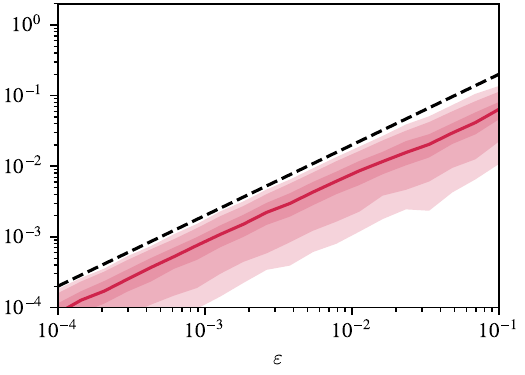}\\
		&&&\\
		\raisebox{1.1cm}{\multirow{2}{*}{\rotatebox{90}{Delaunay--Rips}}} &
		\raisebox{.3cm}{\scalebox{.55}{\rotatebox{90}{$\bottle(\dgmdr{1}(X),\dgmdr{1}(Y))$}}} &
		\hspace{-1mm}\includegraphics[width=.25\linewidth]{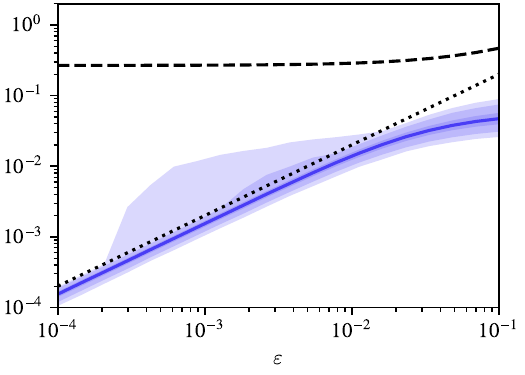}&
		\hspace{-2mm}\includegraphics[width=.25\linewidth]{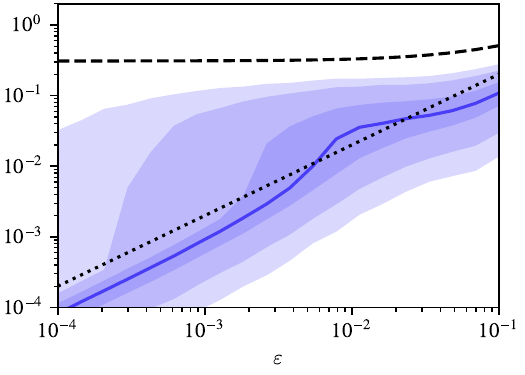}&
		\hspace{-2mm}\includegraphics[width=.25\linewidth]{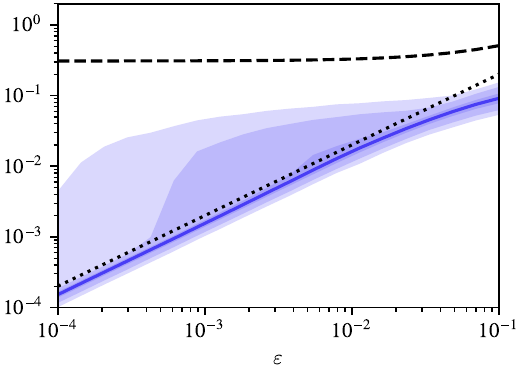}\\
		&
		\raisebox{.3cm}{\scalebox{.55}{\rotatebox{90}{$\bottle(\dgmdr{2}(X),\dgmdr{2}(Y))$}}} &
		\hspace{-1mm}\includegraphics[width=.25\linewidth]{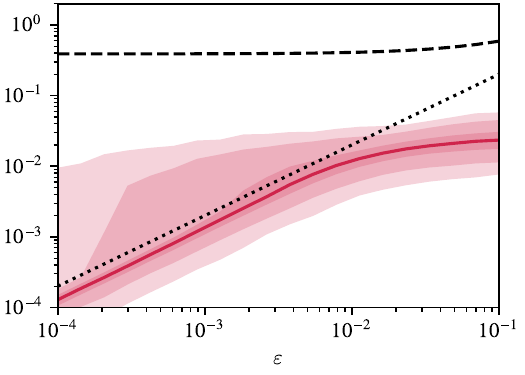}&
		\hspace{-2mm}\includegraphics[width=.25\linewidth]{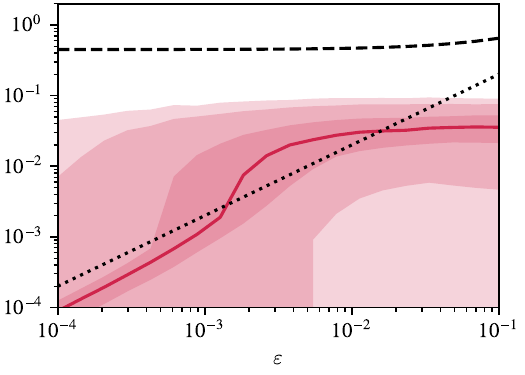}&
		\hspace{-2mm}\includegraphics[width=.25\linewidth]{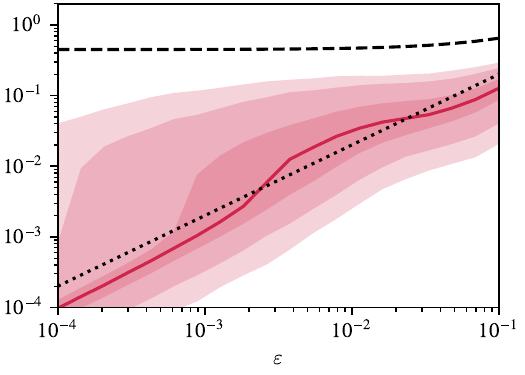}
	\end{tabular}
	\caption{Distribution of the bottleneck distance between Rips or Delaunay--Rips 1- and 2-dimensional persistence diagrams of two point clouds $X$ and $Y\subset\bbr^3$ where $\hausdo(X,Y)\leq\varepsilon$. 
	The left column corresponds to $\pc$ uniformly sampled in $[-1,1]^3$, the center one to $\pc$ sampled near the unit circle, and the right one to $\pc$ sampled near the unit sphere. $Y$ is obtained as a perturbation of $X$ (each point moves by a distance $\leq\varepsilon$).
	The dashed black line corresponds to \cref{eq:bottle_stab2} for the Delaunay--Rips filtration charts, and to \cref{eq:ripsstab} (i.e., $\bottle\leq2\varepsilon$) for the Rips filtration charts. The distributions are represented by their median (solid lines) and their 0.01-, 0.05- and 0.25-quantiles (areas with three color shades).}
	\label{fig:hist_stability}
\end{figure}

\section{Persistent generators example}

We also implemented, for $\ddim=3$, the recovery of \PH{1} and \PH{2} persistent generators, which permits us to have topological representations of each significant handle and cavity in the point cloud.
More precisely, the generator associated with a persistent pair $(\sigma,\tau)$ is the set of simplices in the column of the cell associated with $\tau$ after the reduction.
The 1-generators (resp. 2-generators) consist of $\ug\cup\nmj{1}$ edges (resp. $\usc{2}$ triangles) only.
While these generators are topologically correct, note that nothing encourages them to have a good geometric quality.
\cref{fig:generators} shows \PH{1} and \PH{2} generators computed for points sampled on surfaces in $\bbr^3$ of genus 0, 1, and 3.

\section{Number of simplices}
\label{appendix:number_of_simplices}

\cref{table:simplices_number} gives the number of non-zero persistence pairs, the number of $\hdim$-cells, and the size of $\msa{\hdim}$, $\usc{\hdim}$ and $\del^\hdim$, with various ambient dimensions $\ddim\in\{3,4,5\}$ and point cloud distributions (uniformly in $[-1,1]^\ddim$ or near the sphere $\bbs^{\ddim-1}$).

\begin{table}
	\centering
\begin{tabular}{|r|r|r|r||r|r|r|r|r||r|r|}
\hline
$\ddim$ &
  $n$ &
  distrib. &
  $\hdim$ &
  $|\PH{\hdim-1}|$ &
  $|\poly{\hdim}|$ &
  $|\msa{\hdim}|$ &
  $|\usc{\hdim}|$ &
  $|\del^{\hdim}|$ &
  $\frac{|\poly{\hdim}|}{|\msa{\hdim}|}$ &
  $\frac{|\poly{\hdim}|}{|\del^{\hdim}|}$ \\
\hline
\multirow{6}{*}{3} & \multirow{6}{*}{\rotatebox{90}{10000}} & \multirow{3}{*}{uniform} & 1 & 9999 & 9999  & 9999  & 15609 & 76690  & 1.00 & 0.13 \\
                   &                        &                          & 2 & 5610 & 24783 & 66691 & 71315 & 133258 & 0.37 & 0.19 \\
                   &                        &                          & 3 & 2938 & 4624  &     / &     / & 66567  &    / & 0.07 \\
\cline{3-11}
                   &                        & \multirow{3}{*}{sphere}  & 1 & 9999 & 9999  & 9999  & 12827 & 64045  & 1.00 & 0.16 \\
                   &                        &                          & 2 & 2828 & 18570 & 54046 & 54379 & 103762 & 0.34 & 0.18 \\
                   &                        &                          & 3 & 174  & 333   &     / &     / & 49716  &    / & 0.01 \\
\hline
\multirow{8}{*}{4} & \multirow{8}{*}{\rotatebox{90}{1000}}  & \multirow{4}{*}{uniform} & 1 & 999  & 999   & 999   & 1713  & 16109  & 1.00 & 0.06 \\
                   &                        &                          & 2 & 714  & 5517  & 15110 & 16942 & 54026  & 0.37 & 0.10 \\
                   &                        &                          & 3 & 433  & 15218 & 38916 & 40455 & 64547  & 0.39 & 0.24 \\
                   &                        &                          & 4 & 297  & 1539  &     / &     / & 25631  &    / & 0.06 \\
\cline{3-11}
                   &                        & \multirow{4}{*}{sphere}  & 1 & 999  & 999   & 999   & 1539  & 11549  & 1.00 & 0.09 \\
                   &                        &                          & 2 & 540  & 3872  & 10550 & 11386 & 32686  & 0.37 & 0.12 \\
                   &                        &                          & 3 & 277  & 8305  & 22136 & 22592 & 34720  & 0.38 & 0.24 \\
                   &                        &                          & 4 & 26   & 456   &     / &     / & 12584  &    / & 0.04 \\
\hline
\multirow{10}{*}{5} & \multirow{10}{*}{\rotatebox{90}{100}}  & \multirow{5}{*}{uniform} & 1 & 99   & 99    & 99    & 161   & 1922   & 1.00 & 0.05 \\
                   &                        &                          & 2 & 62   & 617   & 1823  & 2942  & 9812   & 0.34 & 0.06 \\
                   &                        &                          & 3 & 26   & 3793  & 7989  & 8513  & 19935  & 0.47 & 0.19 \\
                   &                        &                          & 4 & 21   & 4477  & 11946 & 12229 & 17652  & 0.37 & 0.25 \\
                   &                        &                          & 5 & 10   & 283   &     / &     / & 5706   &    / & 0.05 \\
\cline{3-11}   
                   &                        & \multirow{5}{*}{sphere}  & 1 & 99   & 99    & 99    & 163   & 1561   & 1.00 & 0.06 \\
                   &                        &                          & 2 & 64   & 523   & 1462  & 2209  & 6837   & 0.36 & 0.08 \\
                   &                        &                          & 3 & 42   & 2688  & 5375  & 5747  & 12441  & 0.50 & 0.22 \\
                   &                        &                          & 4 & 37   & 2801  & 7066  & 7209  & 10121  & 0.40 & 0.28 \\
                   &                        &                          & 5 & 4    & 143   &    /  &    /  & 3055   &    / & 0.05 \\
\hline
\end{tabular}
	\caption{Number of non-zero persistence pairs and size of the different 
	constructions for various ambient dimensions and point cloud distributions. 
	$|\PH{\hdim-1}|$ denotes the number of non-zero $(\hdim-1)$-dimensional 
	persistence pairs. $|\poly{\hdim}|$ denotes the number of $\hdim$-cells 
	delimited by $\usc{\hdim-1}\cup\nmj{\hdim-1}$ within $\msa{\hdim}$.}
	\label{table:simplices_number}
\end{table}

\setcounter{fnsave}{\value{footnote}}
\begin{figure}
	\begin{center}
	\begin{tabular}{ccccc}
		& Surface & Input point cloud & 1-generators & 2-generators\\
		\hspace{-.2cm}\rotatebox{90}{\hspace{.3cm}comet,\figfootnotemark\ $g=0$}&
		\hspace{-.3cm}\includegraphics[width=.235\linewidth]{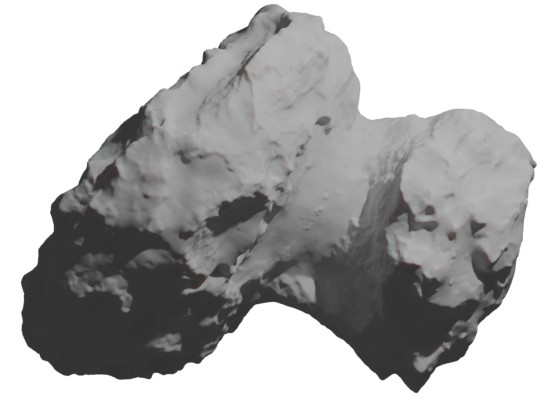}&
		\hspace{-.3cm}\includegraphics[width=.235\linewidth]{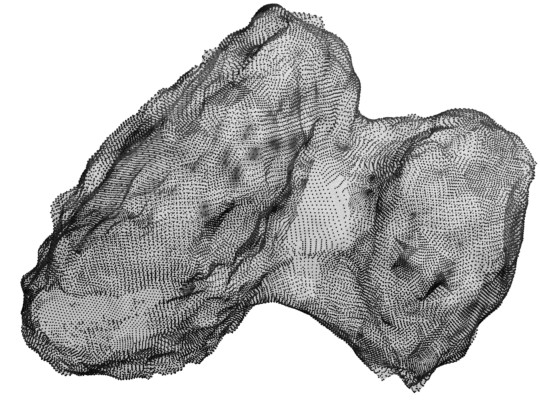}&
		\hspace{-.3cm}\includegraphics[width=.235\linewidth]{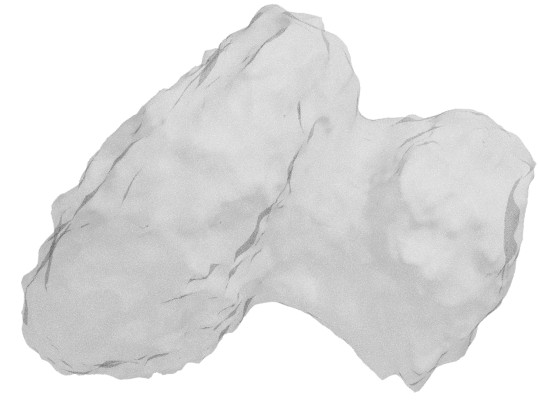}&
		\hspace{-.3cm}\includegraphics[width=.235\linewidth]{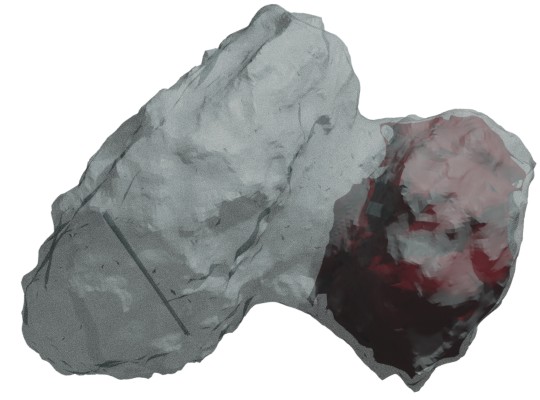}\\
		\hspace{-.2cm}\rotatebox{90}{\hspace{.6cm}Hilbert,\figfootnotemark\ $g=1$}&
		\hspace{-.3cm}\includegraphics[width=.235\linewidth]{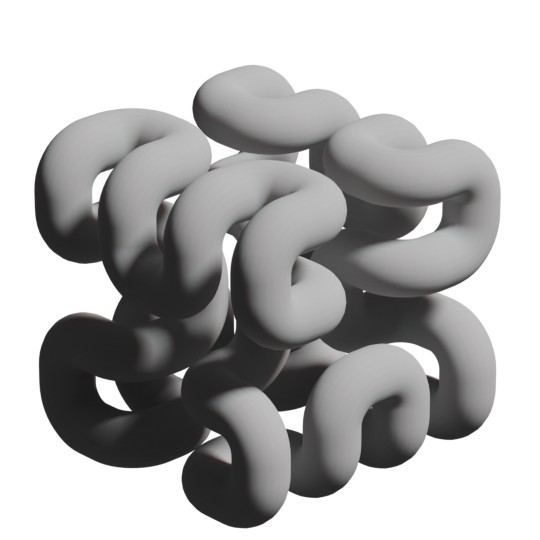}&
		\hspace{-.3cm}\includegraphics[width=.235\linewidth]{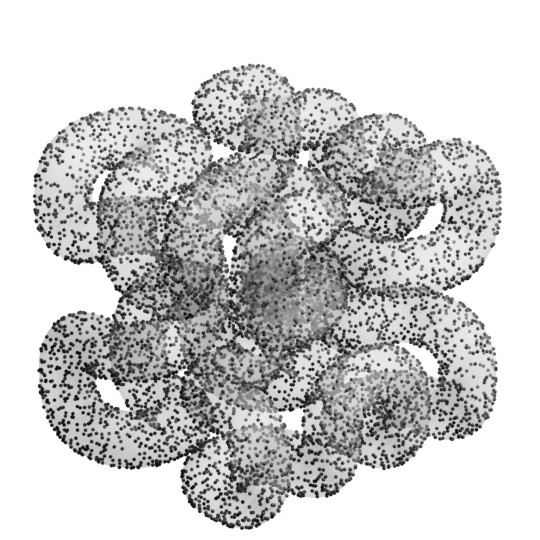}&
		\hspace{-.3cm}\includegraphics[width=.235\linewidth]{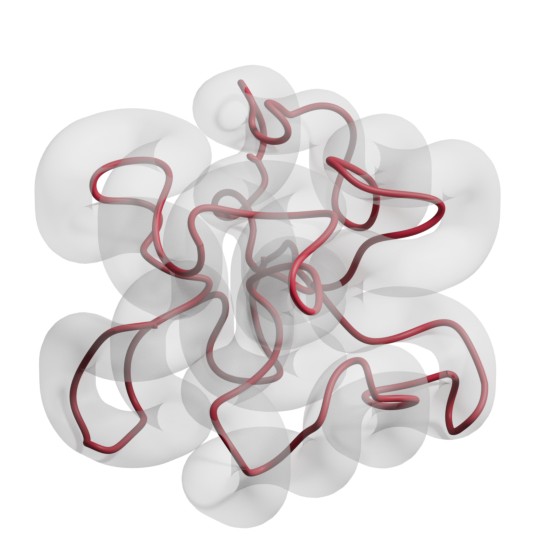}&
		\hspace{-.3cm}\includegraphics[width=.235\linewidth]{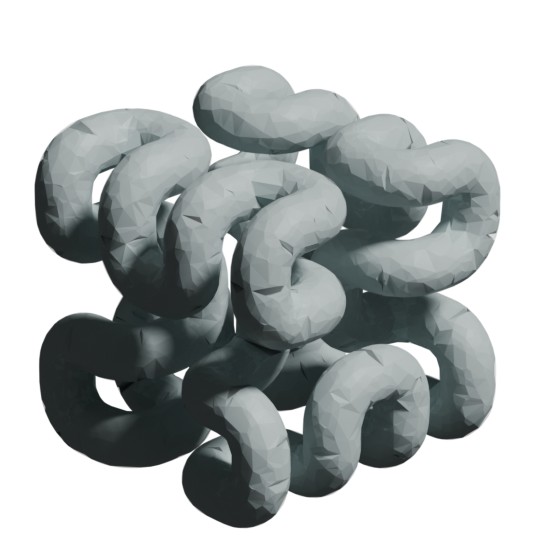}\\
		\hspace{-.2cm}\rotatebox{90}{\hspace{.2cm}lion,\figfootnotemark\ $g=3$}&
		\hspace{-.3cm}\includegraphics[width=.235\linewidth]{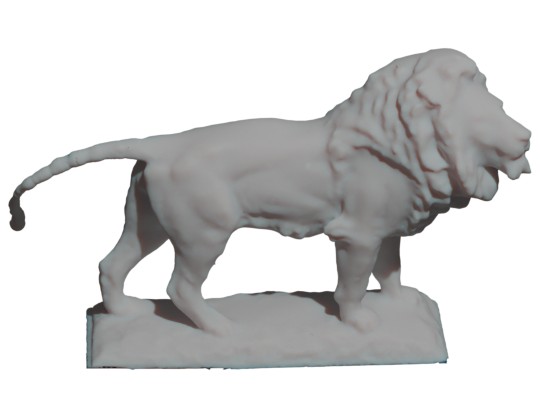}&
		\hspace{-.3cm}\includegraphics[width=.235\linewidth]{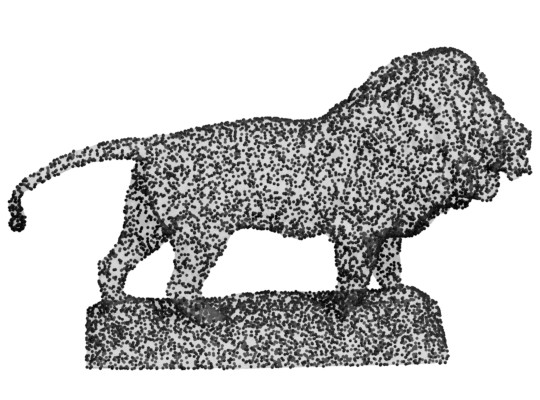}&
		\hspace{-.3cm}\includegraphics[width=.235\linewidth]{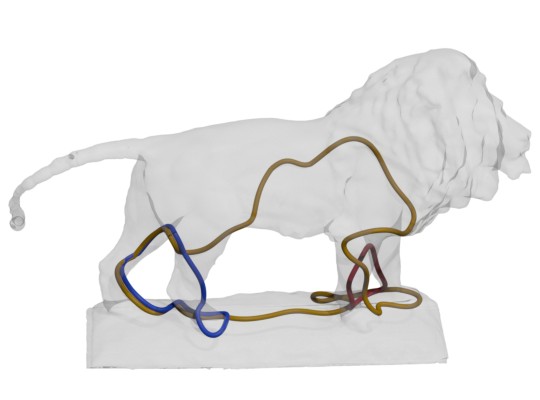}&
		\hspace{-.3cm}\includegraphics[width=.235\linewidth]{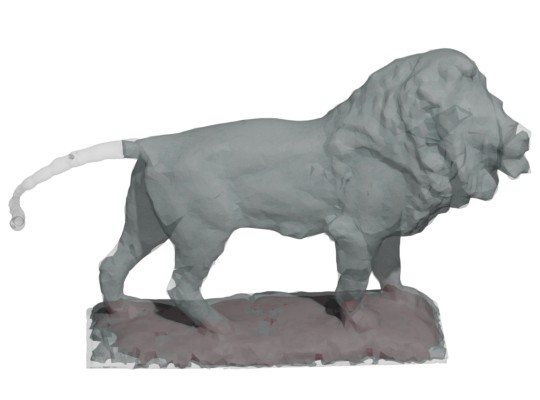}
	\end{tabular}
	\end{center}
	\caption{1- and 2-dimensional persistent generators for point clouds sampled on surfaces of genus 0 (comet, top row), genus 1 (Hilbert curve, middle row), and genus 3 (lion statue, bottom row).
	Only persistent generators with significant persistence are shown.
	Only the largest connected component of each 1-generator was kept and was smoothed for readability.
	1- and 2-generators can be interpreted respectively as the topological handles and topological cavities of the input point cloud.
	}
	\label{fig:generators}
\end{figure}
\setcounter{footnote}{\value{fnsave}}
\figfootnotetext{\url{https://commons.wikimedia.org/wiki/File:67P-Churyumov-Gerasimenko.stl}}
\figfootnotetext{\url{https://ten-thousand-models.appspot.com/detail.html?file_id=53754}}
\figfootnotetext{\url{https://ten-thousand-models.appspot.com/detail.html?file_id=77319}}

\section{Additional results on time series classification}

\cref{table:ucr_scores} contains the classification accuracies for the UCR time series datasets using the Delaunay--Rips and Delaunay--\cech\ filtrations. \cref{table:ucr_datasets} contains a description of these datasets (train and test sizes, length of the time series, number of classes). 

\begin{table}
	\centering
	\begin{tabular}{|r|rrr|rrr|}
	\hline
	& \multicolumn{3}{c|}{Delaunay--Rips} & \multicolumn{3}{c|}{Delaunay--\cech}\\
	& ATOL & PL & stats & ATOL & PL & stats \\
	\hline
	\texttt{DPOAG} & \textbf{0.706} & 0.676 & 0.669 & 0.701 & 0.655 & 0.705\\
	\texttt{DPOC} & 0.696 & 0.696 & 0.692 & 0.703 & 0.674 & \textbf{0.717}\\
	\texttt{DPTW} & 0.589 & 0.612 & \textbf{0.662} & 0.618 & 0.590 & 0.597\\
	\texttt{PPOAG} & 0.811 & 0.810 & 0.824 & \textbf{0.834} & 0.790 & 0.815\\
	\texttt{PPOC} & \textbf{0.831} & 0.784 & 0.790 & 0.813 & 0.770 & 0.790\\
	\texttt{PPTW} & 0.734 & \textbf{0.766} & 0.717 & 0.732 & \textbf{0.766} & 0.761\\
	\texttt{ECG200} & 0.724 & 0.720 & 0.730 & 0.724 & 0.730 & \textbf{0.780}\\
	\texttt{IPD} & 0.708 & 0.700 & \textbf{0.751} & 0.724 & 0.527 & 0.746\\
	\texttt{MI} & 0.589 & 0.524 & 0.633 & 0.588 & 0.480 & \textbf{0.636}\\
	\texttt{Plane} & 0.855 & 0.829 & \textbf{0.933} & 0.850 & 0.857 & 0.924\\
	\texttt{SL} & 0.701 & 0.622 & 0.771 & 0.722 & 0.627 & \textbf{0.784}\\
	\texttt{GP} & 0.811 & 0.800 & 0.887 & 0.833 & 0.773 & \textbf{0.913}\\
	\texttt{GPAS} & 0.934 & 0.858 & \textbf{0.943} & 0.941 & 0.848 & 0.937\\
	\texttt{GPMVF} & 0.931 & 0.839 & 0.943 & \textbf{0.945} & 0.804 & 0.930\\
	\texttt{GPOVY} & 0.993 & 0.986 & \textbf{1.000} & 0.993 & 0.959 & \textbf{1.000}\\
	\texttt{PC} & 0.880 & 0.806 & \textbf{0.900} & 0.872 & 0.828 & 0.872\\
	\texttt{SC} & 0.523 & 0.487 & 0.550 & 0.547 & 0.500 & \textbf{0.553}\\
	\texttt{EQ} & 0.733 & 0.741 & 0.741 & \textbf{0.745} & 0.655 & 0.741\\
	\texttt{LN2} & \textbf{0.674} & 0.672 & 0.623 & 0.669 & 0.574 & 0.590\\
	\texttt{Worms} & 0.699 & 0.558 & 0.727 & 0.700 & 0.623 & \textbf{0.792}\\
	\texttt{Haptics} & 0.383 & 0.370 & 0.370 & \textbf{0.387} & 0.192 & 0.321\\
	\texttt{ACSF1} & 0.686 & 0.550 & \textbf{0.780} & 0.707 & 0.320 & 0.690\\
	\texttt{PigCVP} & 0.140 & 0.019 & 0.130 & 0.125 & 0.019 & \textbf{0.173}\\
	\texttt{PArtP} & 0.427 & 0.240 & \textbf{0.476} & 0.429 & 0.149 & 0.452\\
	\texttt{PAirP} & 0.236 & 0.159 & 0.255 & 0.254 & 0.149 & \textbf{0.332}\\
	\texttt{HT} & 0.714 & 0.761 & \textbf{0.933} & 0.779 & 0.676 & 0.924\\
	\texttt{Rock} & \textbf{0.420} & \textbf{0.420} & \textbf{0.420} & \textbf{0.420} & \textbf{0.420} & \textbf{0.420}\\
	\texttt{KLC} & 0.838 & 0.476 & \textbf{0.882} & 0.853 & 0.444 & 0.870\\
	\texttt{CD} & 0.574 & 0.604 & \textbf{0.652} & 0.575 & 0.555 & 0.622\\
	\texttt{BH} & 0.702 & 0.678 & \textbf{0.722} & 0.693 & 0.683 & 0.717\\
	\texttt{DAG} & 0.180 & 0.200 & \textbf{0.300} & 0.100 & 0.180 & 0.240\\
	\hline
	\end{tabular}
	\caption{Classification accuracies on the UCR time series datasets.}
	\label{table:ucr_scores}
\end{table}

\begin{table}
  \centering
  \begin{tabular}{|l|l|rrrrl|}
  \hline
  Dataset   & UCR dataset full name          &   \rotatebox{90}{Train} &   \rotatebox{90}{Test} &   \rotatebox{90}{Length} &   \rotatebox{90}{Classes} & Category     \\
  \hline
  \texttt{DPOAG}     & DistalPhalanxOutlineAgeGroup   &     400 &    139 &       80 &         3 & image        \\
  \texttt{DPOC}      & DistalPhalanxOutlineCorrect    &     600 &    276 &       80 &         2 & image        \\
  \texttt{DPTW}      & DistalPhalanxTW                &     400 &    139 &       80 &         6 & image        \\
  \texttt{PPOAG}     & ProximalPhalanxOutlineAgeGroup &     400 &    205 &       80 &         3 & image        \\
  \texttt{PPOC}      & ProximalPhalanxOutlineCorrect  &     600 &    291 &       80 &         2 & image        \\
  \texttt{PPTW}      & ProximalPhalanxTW              &     400 &    205 &       80 &         6 & image        \\
  \texttt{ECG200}    & ECG200                         &     100 &    100 &       96 &         2 & ecg          \\
  \texttt{IPD}       & ItalyPowerDemand               &      67 &   1029 &       24 &         2 & sensor       \\
  \texttt{MI}        & MedicalImages                  &     381 &    760 &       99 &        10 & image        \\
  \texttt{Plane}     & Plane                          &     105 &    105 &      144 &         7 & sensor       \\
  \texttt{SL}        & SwedishLeaf                    &     500 &    625 &      128 &        15 & image        \\
  \texttt{GP}        & GunPoint                       &      50 &    150 &      150 &         2 & har          \\
  \texttt{GPAS}      & GunPointAgeSpan                &     135 &    316 &      150 &         2 & har          \\
  \texttt{GPMVF}     & GunPointMaleVersusFemale       &     135 &    316 &      150 &         2 & har          \\
  \texttt{GPOVY}     & GunPointOldVersusYoung         &     136 &    315 &      150 &         2 & har          \\
  \texttt{PC}        & PowerCons                      &     180 &    180 &      144 &         2 & device       \\
  \texttt{SC}        & SyntheticControl               &     300 &    300 &       60 &         6 & simulated    \\
  \texttt{EQ}        & Earthquakes                    &     322 &    139 &      512 &         2 & sensor       \\
  \texttt{LN2}       & Lightning2                     &      60 &     61 &      637 &         2 & sensor       \\
  \texttt{Worms}     & Worms                          &     181 &     77 &      900 &         5 & motion       \\
  \texttt{Haptics}   & Haptics                        &     155 &    308 &     1092 &         5 & motion       \\
  \texttt{ACSF1}     & ACSF1                          &     100 &    100 &     1460 &        10 & device       \\
  \texttt{PigCVP}    & PigCVP                         &     104 &    208 &     2000 &        52 & hemodyn.     \\
  \texttt{PArtP}     & PigArtPressure                 &     104 &    208 &     2000 &        52 & hemodyn.     \\
  \texttt{PAirP}     & PigAirwayPressure              &     104 &    208 &     2000 &        52 & hemodyn.     \\
  \texttt{HT}        & HouseTwenty                    &      40 &    119 &     2000 &         2 & device       \\
  \texttt{Rock}      & Rock                           &      20 &     50 &     2844 &         4 & spectro      \\
  \texttt{KLC}       & KeplerLightCurves              &     920 &    399 &     4767 &         7 & sensor       \\
  \texttt{CD}        & CatsDogs                       &     164 &    164 &    14773 &         2 & audio        \\
  \texttt{BH}        & BinaryHeartbeat                &     204 &    205 &    18530 &         2 & audio        \\
  \texttt{DAG}       & DucksAndGeese                  &      50 &     50 &   236784 &         5 & audio        \\
  \hline
  \end{tabular}
  \caption[Description of the UCR datasets]{The list of the UCR time series datasets used in \cref{sec:classification}, with their acronym and their full name, the size of train and test sets, the length of the time series and the number of classes.}
  \label{table:ucr_datasets}
\end{table}

\end{document}